\documentclass[11pt,a4paper]{article}

\usepackage{amsfonts,amsmath,amssymb,amsthm}
\usepackage{latexsym,amscd}
\usepackage{amsbsy}
\usepackage{CJK}
\usepackage{indentfirst,mathrsfs}
\usepackage{amsfonts,amsmath,amssymb,amsthm}
\usepackage{latexsym,amscd}
\usepackage{amsbsy}
\usepackage{color}
\usepackage{multirow}
\usepackage{makecell}
\usepackage{float}

\newtheorem{thm}{Theorem}
\newtheorem{lem}{Lemma}

\newtheorem{example}{Example}
\newtheorem{remark}{Remark}

\newcommand{\fp}{{\mathbb F}_{p}}
\newcommand{\fq}{{\mathbb F}_{q}}
\newcommand{\ftwo}{{\mathbb F}_{2}}
\newcommand{\fqm}{{\mathbb F}_{q^m}}
\newcommand{\ftwom}{{\mathbb F}_{2^m}}

\newcommand{\fqk}{{\mathbb F}_{q^k}}
\newcommand{\fqr}{{\mathbb F}_{q^r}}
\newcommand{\fqs}{{\mathbb F}_{q^s}}

\newcommand{\Tr}{{\rm {Tr}}}
\newcommand{\C}{{\mathcal{C}}}

\begin{document}

\title{A subfield-based construction of optimal linear codes \\ over finite fields}
\author{Zhao Hu, Nian Li, Xiangyong Zeng, Lisha Wang, Xiaohu Tang
\thanks{The authors are with the Hubei Key Laboratory of Applied Mathematics, Faculty of Mathematics and Statistics, Hubei
University, Wuhan, 430062, China. Xiaohu Tang is also with the Information Security and National Computing Grid Laboratory, Southwest Jiaotong University, Chengdu, 610031, China. Email: zhao.hu@aliyun.com, nian.li@hubu.edu.cn, xzeng@hubu.edu.cn, wangtaolisha@163.com,xhutang@swjtu.edu.cn}
}
\date{}
\maketitle
\begin{quote}
  {\small {\bf Abstract:} In this paper, we construct four families of linear codes over finite fields from the complements of either the union of subfields or the union of cosets of a subfield, which can produce infinite families of optimal linear codes, including infinite families of (near) Griesmer codes. We also characterize the optimality of these four families of linear codes with an explicit computable criterion using the Griesmer bound and obtain many distance-optimal linear codes. In addition, we obtain several classes of distance-optimal linear codes with few weights and completely determine their weight distributions. It is shown that most of our linear codes are self-orthogonal or minimal which are useful in applications.}

\end{quote}

\section{Introduction}

Let $\mathbb{F}_{q^m}$ be the finite field with $q^m$ elements and $\mathbb{F}_{q^m}^{*}=\mathbb{F}_{q^m}\backslash\{0\}$, where $q$ is a power of a prime $p$ and $m$ is a positive integer. An $[n, k, d]$ linear code $\mathcal{C}$ over $\fq$ is a $k$-dimensional subspace of $\fq^{n}$ with minimum Hamming distance $d$. An $[n,k,d]$ linear code $\mathcal{C}$ over $\fq$ is said to be distance-optimal if no $[n,k,d+1]$ code exists (i.e., this code has the largest minimum distance for given length $n$ and dimension $k$) and it is called almost distance-optimal if there exists an $[n,k,d+1]$ distance-optimal code. An $[n,k,d]$ linear code $\mathcal{C}$ is called optimal if its parameters $n$, $k$ and $d$ meet a bound on linear codes with equality and almost optimal if its parameters $n$, $k$ and $d+1$ meet a bound on linear codes with equality \cite{HWPV}. Optimal linear codes are  important in both theory and practice, the reader is referred to \cite{HDWZ,JYHLL,WZQY} for recent results. The Griesmer bound \cite{JHG,GSJS} for an $[n,k,d]$ linear code $\C$ over $\fq$ is given by
\[ n\geq  g(k,d):=\sum_{i=0}^{k-1} \lceil \frac{d}{q^i}\rceil,\]
where $\lceil \cdot \rceil$ denotes the ceiling function. An $[n,k,d]$ linear code $\mathcal{C}$ is called a Griesmer code if its parameters $n$, $k$ and $d$ achieve the Griesmer bound and called a near Griesmer code if $n-1$, $k$ and $d$ achieve the Griesmer bound. Griesmer codes have been an interesting topic of study for many years due to not only their optimality but also their geometric applications \cite{DING1,DING2}.

The dual code of an $[n,k,d]$ linear code $\C$ over $\fq$ is defined by
$ \C^{\bot}=\{x\in \fq^{n}\,\,|\,\, x \cdot y = 0 \,\, {\rm for \,\, all} \,\, y\in \C \},$
where $x \cdot y$ denotes the Euclidean inner product of $x$ and $y$.
The code $\C^{\bot}$ is an $[n,n-k]$ linear code over $\fq$. A linear code $\C$ is called projective if its dual code has minimum distance at least $3$, and it is called self-orthogonal if $\C\subseteq \C^{\bot}$.
Let $A_{i}$ denote the number of codewords with Hamming weight $i$ in a code $\mathcal{C}$ of length $n$. The weight enumerator of $\mathcal{C}$ is defined by
$1+A_{1}z+A_{2}z^{2}+\cdots +A_{n}z^{n}$. The sequence $(1, A_{1}, A_{2}, \cdots ,A_{n})$ is called the weight distribution of $\mathcal{C}$.
A code is said to be a $t$-weight code if the number of nonzero $A_{i}$ in the sequence $(A_{1}, A_{2}, \cdots ,A_{n})$ is equal to $t$. Linear codes with few weights have applications in secret sharing schemes \cite{ARJD,CCDY}, authentication codes \cite{DCHT,DW}, association schemes \cite{CAGJ}, strongly regular graphs and some other fields.

In 2007, Ding and Niederreiter \cite{DN} introduced a nice and generic way to construct linear codes via trace functions. Let $D \subset \fqm$ and define
\begin{equation} \label{CD}
\C_D=\{c_{a}=(\Tr_{q}^{q^{m}}(ax))_{x\in D}: a\in \fqm\},
\end{equation}
where $\Tr_{q}^{q^{m}}(\cdot)$ is the trace function from $\mathbb{F}_{q^m}$ to $\fq$.
Then $\C_{D}$ is a linear code of length $n:=|D|$ over $\mathbb{F}_{q}$. The set $D$ is called the defining set of $\mathcal{C}_D$. Later, Ding and Niederreiter's construction was extended to the bivariate form, namely, linear codes of the form
\begin{equation} \label{CDbi}
\mathcal{C}_{D}=\{c_{a,b}=(\Tr_{q}^{q^{m}}(ax)+\Tr_{q}^{q^{k}}(by))_{(x,y) \in D}: a\in \fqm,b \in \fqk\},
\end{equation}
where $D \subset \fqm \times \fqk$.

The objective of this paper is to construct (distance-) optimal linear codes and (near) Griesmer codes over the finite field $\fq$ of the forms \eqref{CD} and \eqref{CDbi}.  Our main contributions are summarized as follows:
\begin{enumerate}
  \item [1)] we construct distance-optimal linear codes $\C_{D}$ of the form \eqref{CD} with the defining set
      $D=\fqm\backslash \Omega_1$, where $\Omega_1=\cup_{i=1}^{h}\mathbb{F}_{q^{r_{i}}}$ and $1\leq r_{1} < r_{2} < \cdots < r_{h}<m$. A criterion $\sum_{i=1}^{h}q^{r_{i}}-|\Omega_1|<r_{1}+h-1$ for $\C_{D}$ to be distance-optimal is given by using the Griesmer bound,  which enables us to obtain many distance-optimal linear codes. In particular, when $h=1$, our construction reduces to the Solomon and Stiffler codes in the nonprojective case (see \cite{HELL,GSJS}), and when $h=2$, we show that the code $\C_{D}$ is a near Griesmer code if $(q,t)=(2,1)$ and distance-optimal if $r_{1}+1>q^{t}$, where $\gcd(r_{1},r_{2})=t$. Further, when $h=1$ and $h=2$, the weight distributions of $\C_{D}$ are completely determined, which are $2$-weight and $5$-weight respectively.
  \item [2)] we construct Griesmer codes and distance-optimal linear codes $\C_{D}$ of the form \eqref{CD} with the defining set $D=\fqm\backslash \Omega_2$, where $\Omega_2=\cup_{i=0}^{h}(\theta_{i}+\fqr)$, $r|m$, $\theta_{0}=0$ and $\theta_{i}\in\fqm^{*}$ for any $1\leq i \leq h$. This construction produces Griesmer codes if $h+1\leq q$ which have different parameters with the Solomon and Stiffler codes in the nonprojective case if $h\ne 0$. When $h+1> q$, we give an explicit computable criterion on $\C_{D}$ such that it is distance-optimal and consequently obtain many distance-optimal linear codes. It is proved that in this construction $\C_{D}$ is at most $(h+2)$-weight and the weight distributions of $\C_{D}$ for $h=1$ and $h=2$ are completely determined which are $3$-weight and $4$-weight respectively.
  \item [3)] we characterize the optimality of the linear codes $\C_{D}$ of the form \eqref{CD} with the defining set $D=\fqm\backslash \Omega_3$, where $\Omega_3=\cup_{i=1}^{h}(\theta_{i}*\fqr)$, $r|m$, $\theta_{i}\in \fqm^{*}$ for any $1\leq i \leq h$, and give an explicit computable criterion on $\C_{D}$ such that it is distance-optimal. This allows us to produce many distance-optimal linear codes from this construction.
      It is shown that $\C_{D}$ is a Griesmer code with the same parameters as the Solomon and Stiffler code in the nonprojective case if $h=1$ and it is a near Griesmer code if $h=2$ or $(q,h)=(2,3)$. In addition, we prove that $\C_{D}$ is at most $(h+1)$-weight and completely determine its weight distributions for $h=2$ and $h=3$, which are $3$-weight and $3$ or $4$-weight respectively.

  \item [4)] we characterize the optimality of the linear codes $\C_{D}$ of the form \eqref{CDbi} for the defining set $D= \{(x,y):x\in \fqm\backslash \fqr, y \in \fqk\backslash \fqs \}$, where $r|m$ and $s|k$, and give an explicit computable criterion on $\C_{D}$ such that it is distance-optimal. As a consequence, we obtain many distance-optimal linear codes from this construction. In addition, the weight distribution of $\C_{D}$ is completely determined which is shown to be $4$-weight. A similar discussion for $r=s=0$ and ${\mathbb F}_{q^0}=\{0\}$ shows that the linear code $\C_{D}$ is a near Griesmer code (distance-optimal if $m+\lfloor \frac{2}{q}\rfloor>1$) when $m=k$ and it is a Griesmer code with the same parameters as the Solomon and Stiffler code in the nonprojective case when $m\ne k$.
\end{enumerate}

We also investigate the self-orthogonality and minimality of the linear codes constructed in this paper. Self-orthogonal codes can be used to construct quantum error-correcting codes \cite{KKKS} which can protect quantum information in quantum computations and quantum communications \cite{CRSS, DGOT} and minimal linear codes can be used to construct secret sharing schemes \cite{CCDY,DY,JLM,JYCD}. It is shown that most of our linear codes are either self-orthogonal or minimal.

\section{Preliminaries}

Let $q$ be a power of a prime $p$ and denote the canonical additive character of $\fq$ by
\[\chi(x)=\zeta_{p}^{\Tr_{p}^{q}(x)},\]
where $\zeta_{p}$ is a primitive complex $p$-th root of unity and $\Tr_{p}^{q}(\cdot)$ is the trace function from $\fq$ to $\fp$.

\begin{lem}(\cite{Lidl})\label{lem-trace}
 Let $\alpha\in \fqm$. Then $\Tr_{q}^{q^{m}}(\alpha)=0$ if and only if $\alpha=\beta^q-\beta$ for some $\beta\in \fqm$.
\end{lem}

The optimality of near Griesmer codes can be determined as below.

\begin{lem}\label{near-Griesmer}
Let $\C$ be an $[n,k,d]$ near Griesmer code over $\fq$ with $k>1$. Then $\C$ is distance-optimal if $q\mid d$ and almost distance-optimal if $q\nmid d$.
\end{lem}
\begin{proof}
According to the Griesmer bound one obtains $n=g(k,d)+1=\sum_{i=0}^{k-1} \lceil \frac{d}{q^i}\rceil+1$.
To complete the proof, it suffices to prove that
$g(k,d+1)>g(k,d)+1$ if $q\mid d$ and $g(k,d+1)=g(k,d)+1$ if $q\nmid d$. Note that
\begin{align*}
g(k,d+1)-g(k,d)-1&=\sum_{i=0}^{k-1} \lceil \frac{d+1}{q^i}\rceil-\sum_{i=0}^{k-1} \lceil \frac{d}{q^i}\rceil-1
=\sum_{i=1}^{k-1} \lceil \frac{d+1}{q^i}\rceil-\sum_{i=1}^{k-1} \lceil \frac{d}{q^i}\rceil.
\end{align*}
Then the result follows from the fact that $\lceil \frac{d+1}{q^i}\rceil=\lceil \frac{d}{q^i}\rceil+1$ if $q^i\mid d$ and otherwise $\lceil \frac{d+1}{q^i}\rceil=\lceil \frac{d}{q^i}\rceil$ for any integer $i>0$.
This completes the proof.
\end{proof}

Some results on the self-orthogonality of linear codes of the form \eqref{CD} are given as follows.

\begin{lem}\label{self-orthogonalF}
Let $q$ be a power of a prime $p$ and $m$, $r$ be positive integers with $r\mid m$ and $(q,r)\notin \{(2,1),(2,2),(3,1)\}$. Define \\
 \indent $1).$ $D= \fqr;$ or\\
 \indent $2).$ $D=\{x+\theta: x\in \fqr\};$ or\\
 \indent $3).$ $D=\{x*\theta: x\in \fqr\}$\\
where $\theta\in \fqm\backslash \fqr$. Then the linear code $\C_{D}$ defined in \eqref{CD} is self-orthogonal.
\end{lem}
\begin{proof}
$1)$ Observe that $\Tr_{q}^{q^{m}}(ax)=\Tr_{q}^{q^{r}}(\Tr_{q^{r}}^{q^{m}}(a)x)$ if $x\in D=\fqr$ which implies that in this case the linear code $\C_{D}$ can be reduced to $\{c_{a}=(\Tr_{q}^{q^{r}}(ax))_{x\in \fqr}: a\in \fqr\}$. Thus, to complete the proof, it is sufficient to show that $c_{a}\cdot c_{b}=0$ for any $a,b\in \fqr^*$. If $a$ and $b$ are linearly dependent over $\fq$, namely, there exists some $u\in \fq^*$ such that $b=ua$, then by the balanced property of trace functions one has that
$$c_{a}\cdot c_{b}=\sum\nolimits_{x\in\fqr}\Tr_{q}^{q^{r}}(ax)\Tr_{q}^{q^{r}}(uax)=u\sum\nolimits_{x\in\fqr}\Tr_{q}^{q^{r}}(ax)^2=uq^{r-1}\sum\nolimits_{y\in\fq}y^2.$$
Let $\alpha$ be a primitive element of $\fq$, it can be readily verified that $\sum_{y\in\fq}y^2=0$ if $\alpha\ne 1$ and $\alpha\ne -1$, i.e., $q\ne 2$ and $q\ne 3$. Then $c_{a}\cdot c_{b}=0$ if $(q,r)\ne (2,1)$ or $(3,1)$. Now assume that $a$ and $b$ are linearly independent over $\fq$, which implies $r\geq 2$. Then, for any $(s,t)\in \fq^2$,  we have
\begin{align*}
N_{s,t}:=&|\{ x\in \fqr: \Tr_{q}^{q^{r}}(ax)=s\,\,{\rm and}\,\, \Tr_{q}^{q^{r}}(bx)=t\}|\\
=&\frac{1}{q^2}\sum_{x\in \fqr}\sum_{u\in \fq}\chi(u(\Tr_{q}^{q^{r}}(ax)-s))
\sum_{v\in \fq}\chi(v(\Tr_{q}^{q^{r}}(bx)-t))\\
=&\frac{1}{q^2}\sum_{u\in \fq}\sum_{v\in \fq}\chi(-(us+vt))
\sum_{x\in \fqr}\chi(\Tr_{q}^{q^{r}}((au+bv)x))\\
=&q^{r-2}
\end{align*}
due to $au+bv=0$ if and only if $u=v=0$. This gives
$$c_{a}\cdot c_{b}=\sum_{x\in\fqr}\Tr_{q}^{q^{r}}(ax)\Tr_{q}^{q^{r}}(bx)=q^{r-2}\sum_{s,t\in \fq}st=q^{r-2}\sum_{s\in \fq}s\sum_{t\in \fq}t=0$$
if $(q,r)\not=(2,2)$.

$2)$ If $D=\{x+\theta: x\in \fqr\}$ for some $\theta\in \fqm\backslash \fqr$, then any codeword $c_{a}\in\C_{D}$ can be expressed as $c_{a}=\bar{c}_{a}+u_{a}$ where $\bar{c}_{a}=(\Tr_{q}^{q^{m}}(ax))_{x\in \fqr}\in \C_{\fqr}$ and $u_{a}=(\Tr_{q}^{q^{m}}(a\theta))_{x\in \fqr}$. For any two codewords $c_{a}=\bar{c}_{a}+u_{a}$ and $c_{b}=\bar{c}_{b}+u_{b}$ in $\C_{D}$, we have
$$c_{a}\cdot c_{b}=(\bar{c}_{a}+u_{a})\cdot(\bar{c}_{b}+u_{b})=\bar{c}_{a}\cdot \bar{c}_{b}+u_{a}\cdot\bar{c}_{b}+u_{b}\cdot\bar{c}_{a}+u_{a}\cdot u_{b}.$$
Note that $u_{a}\cdot u_{b}=q^{r}\Tr_{q}^{q^{m}}(a\theta)\Tr_{q}^{q^{m}}(b\theta)=0$ and
$u_{a}\cdot\bar{c}_{b}=\Tr_{q}^{q^{m}}(a\theta)\sum_{x\in\fqr}\Tr_{q}^{q^{r}}(\Tr_{q^{r}}^{q^{m}}(b)x)$ since $x\in \fqr$.
If $\Tr_{q^{r}}^{q^{m}}(b)=0$, then $u_{a}\cdot\bar{c}_{b}=0$. Otherwise, by the balanced property of trace functions, we have $\sum_{x\in \fqr}\Tr_{q}^{q^{r}}(\Tr_{q^{r}}^{q^{m}}(b)x)=q^{r-1}\sum_{y\in \fq}y=0$ if $(q,r)\ne (2,1)$. Similarly, we have $u_{b}\cdot\bar{c}_{a}=0$ if $(q,r)\ne (2,1)$. Then, by 1) of Lemma \ref{self-orthogonalF}, we have $c_{a}\cdot c_{b}=\bar{c}_{a}\cdot \bar{c}_{b}=0$ if $(q,r)\notin \{(2,1),(2,2),(3,1)\}$.

$3)$ The linear code $\C_{D}$ in this case can be expressed as $\{c_{a}=(\Tr_{q}^{q^{m}}(a\theta x))_{ x\in \fqr}: a\in \fqm\}$ which is exactly the code $\{c_{b}=(\Tr_{q}^{q^{r}}(bx))_{x\in \fqr}: b\in \fqr\}$ since $\Tr_{q}^{q^{m}}(ax)=\Tr_{q}^{q^{r}}(\Tr_{q^{r}}^{q^{m}}(a)x)$ if $x\in \fqr$. Then the result follows from 1) of Lemma \ref{self-orthogonalF}.
\end{proof}

The following lemma can be readily verified  and will be frequently used in the sequel.

\begin{lem}\label{self-orthogonalD-}
Let $D,D_{1},D_{2}\subseteq \fqm$ (resp. $\fqm\times \fqk$).  Then \\
\indent $1).$ Denote $D^{*}=D\backslash \{0\}$ if $0\in D$ (resp. $D^{*}=D\backslash \{(0,0)\}$ if $(0,0)\in D$). $\C_{D}$ defined in \eqref{CD} (resp. \eqref{CDbi}) is self-orthogonal if and only if $\C_{D^{*}}$ defined in \eqref{CD} (resp. \eqref{CDbi}) is self-orthogonal$;$\\
\indent $2).$ Let $D=D_{1}\cup D_{2}$ and $D_{1}\cap D_{2}=\{0\} \,\,{\rm or}\,\, \emptyset$ (resp. $\{(0,0)\} \,\,{\rm or}\,\, \emptyset$). $\C_{D}$ defined in \eqref{CD} (resp. \eqref{CDbi}) is self-orthogonal if $\C_{D_{1}}$ and $\C_{D_{2}}$ defined in \eqref{CD} (resp. \eqref{CDbi}) are both self-orthogonal$;$\\
\indent $3).$ Let $D_{1}\subseteq D_{2}$ and $D=D_{2}\backslash D_{1}$. $\C_{D}$ defined in \eqref{CD} (resp. \eqref{CDbi}) is self-orthogonal if $\C_{D_{1}}$ and $\C_{D_{2}}$ defined in \eqref{CD} (resp. \eqref{CDbi}) are both self-orthogonal.
\end{lem}

A vector $u\in \fq^{n}$ covers a vector $v\in \fq^{n}$ if ${\rm Suppt}(v)\subseteq {\rm Suppt}(u)$, where ${\rm Suppt}(u)= \{1 \leq i \leq n : u_{i}\ne 0\}$ is the support of $u=(u_{1},u_{2},\cdots,u_{n})\in \fq^{n}$. A codeword $u$ in $\C$ is said to be minimal if $u$ covers only the codeword $au$ for all $a \in \fq$, but no other codewords in $\C$. A linear code $\C$ is said to be minimal if every
codeword in $\C$ is minimal.

Aschikhmin and Barg's result is often used to determine whether a linear code is minimal.

\begin{lem} (\cite{AAAB}) \label{minimal}
A linear code $\C$ over $\fq$ is minimal if $w_{min}/w_{max}>(q-1)/q$, where $w_{min}$ and $w_{max}$ denote the minimum and maximum
nonzero Hamming weights in $\C$, respectively.
\end{lem}

\section{The first family of optimal linear codes} \label{section-3}

In this section, we investigate the linear codes $\C_{D}$ of the form \eqref{CD} for
\begin{eqnarray}\label{CD1-D1}
D= \fqm \backslash \Omega_1,\;\;\Omega_1=\cup_{i=1}^{h} {\mathbb F}_{q^{r_{i}}},
\end{eqnarray}
where $m>1$, $1\leq r_{1} < r_{2} < \cdots < r_{h}<m$ are positive integers satisfying $r_{i}|m$ for any $1\leq i \leq h$, $r_{i}\nmid r_{j}$ for any $1\leq i<j \leq h$ and $\gcd(r_{1}, r_{2}, \cdots, r_{h})=t$. In particular, define $t=r_1$ if $h=1$.

For simplicity, define
\begin{eqnarray*}
\Theta_1=\{a \in \fqm: \Tr_{q^{r_{i}}}^{q^{m}}(a)\ne 0 \,\,{\rm for\,\,any\,\,}  i \,\,{\rm and\,\,} \Tr_{q^{\gcd(r_{i},r_{j})}}^{q^{m}}(a)= 0 \,\,{\rm for\,\,any\,\,}
i<j\}.
\end{eqnarray*}

\begin{thm} \label{optimalcode-con-1}
Let $\C_{D}$ be defined by \eqref{CD} and \eqref{CD1-D1}. If $\Theta_1$ is nonempty, then  \\
$1).$ $\C_{D}$ is a $[q^m-|\Omega_1|,m,(q-1)(q^{m-1}-\sum_{i=1}^{h}q^{r_{i}-1})]$ linear code;\\
$2).$ $\C_{D}$ is distance-optimal if $\sum_{i=1}^{h}q^{r_{i}}-|\Omega_1|<r_{1}+h-1$; \\
$3).$ $\C_{D}$ is self-orthogonal if $(q,t)\notin \{(2,1),(2,2),(3,1)\}$;\\
$4).$ $\C_{D}$  is minimal if $q^{m-1}> \sum_{i=1}^{h}q^{r_{i}}$.
\end{thm}

\begin{proof}
Denote $\Upsilon:=\{r_{1},r_{2},\cdots, r_{h}\}$ and $r_{S}:=\gcd(s_{1},s_{2},\cdots ,s_{|S|})$ for any set $S=\{s_{1},s_{2},\cdots ,s_{|S|}\}$, where $s_{i}$'s are positive integers. By the definition of $\Omega_1$ and the principle of inclusion-exclusion, we have
\[|\Omega_1|=\sum_{1\leq i\leq h} |{\mathbb F}_{q^{r_{i}}}|
-\sum_{1\leq i<j\leq h}|{\mathbb F}_{q^{r_{i}}}\cap {\mathbb F}_{q^{r_{j}}}|
+\cdots(-1)^{h-1} |\cap_{i=1}^{h} {\mathbb F}_{q^{r_{i}}}|
=\sum_{\emptyset \not=S\subseteq \Upsilon}(-1)^{|S|-1}q^{r_{S}}\]
and consequently, the length of $\C_{D}$ is $n=q^m-\sum_{\emptyset \not=S\subseteq \Upsilon}(-1)^{|S|-1}q^{r_{S}}$. For any $a\in \fqm^*$, the Hamming weight $wt(c_{a})$ of the codeword $c_{a}$ in $\C_{D}$ is $n-N_{a}$, where $N_{a}=|\{x\in D: \Tr_{q}^{q^{m}}(ax) = 0\}|$. Using the orthogonal property of nontrivial additive characters, for $a\not=0$, we have
\begin{align*}
N_{a}=&\frac{1}{q}\sum_{x\in D}
\sum_{u\in \fq} \chi(u\Tr_{q}^{q^{m}}(ax))
=q^{m-1}-\frac{1}{q}\sum_{x\in \Omega_1}
\sum_{u\in \fq} \chi(u\Tr_{q}^{q^{m}}(ax)):=q^{m-1}-\Delta.
\end{align*}
The above discussions lead to
\begin{align} \label{equation-thm1-1}
wt(c_{a})=(q-1)q^{m-1}-(|\Omega_1|-\Delta).
\end{align}

Next, we determine the maximal value of $|\Omega_1|-\Delta$ for any $a\in \fqm^*$ in order to determine the minimal distance of $\C_{D}$.

To calculate $\Delta$, for any positive integer $l\mid m$, define
\[\Phi({\mathbb F}_{q^{l}}):= \frac{1}{q}\sum_{u\in \fq} \sum_{x\in {\mathbb F}_{q^{l}}}\chi(u\Tr_{q}^{q^{m}}(ax)),\]
which satisfies
\begin{align} \label{equation-thm1-4}
\Phi({\mathbb F}_{q^{l}})
=\left\{\begin{array}{ll}
q^{l}, & \mbox{if}\,\,\Tr_{q^{l}}^{q^{m}}(a)=0,  \\[0.05in]
q^{l-1},  & \mbox{otherwise}.
\end{array}\right.
\end{align}
Utilizing the principle of inclusion-exclusion gives
\begin{align*}
\Delta &=\sum_{1\leq i\leq h} \Phi({\mathbb F}_{q^{r_{i}}})
-\sum_{1\leq i<j\leq h}\Phi({\mathbb F}_{q^{r_{i}}}\cap {\mathbb F}_{q^{r_{j}}})
+\cdots(-1)^{h-1} \Phi(\cap_{i=1}^{h}{\mathbb F}_{q^{r_{i}}})
=\sum_{\emptyset \not=S\subseteq \Upsilon}(-1)^{|S|-1}\Phi({\mathbb F}_{q^{r_{S}}}).
\end{align*}
Then we have
\begin{align} \label{equation-thm1-5}
|\Omega_1|-\Delta=&\sum_{\emptyset \not=S\subseteq \Upsilon}(-1)^{|S|-1}(q^{r_{S}}-\Phi({\mathbb F}_{q^{r_{S}}})):=\sum_{\emptyset \not=S\subseteq \Upsilon}(-1)^{|S|-1}f_{a}(S).
\end{align}
Define $\Upsilon_{i}:=\{\gcd(r_{1},r_{i}),\gcd(r_{2},r_{i}),\cdots, \gcd(r_{i-1},r_{i})\}$.
Then by \eqref{equation-thm1-5} one gets
\begin{align} \label{equation-thm1-2}
|\Omega_1|-\Delta=\sum_{\emptyset \not=S\subseteq \Upsilon}(-1)^{|S|-1}f_{a}(S)
=\sum_{i=1}^{h}f_{a}(\{r_{i}\})-\sum_{i=2}^{h}\sum_{\emptyset \not=S\subseteq \Upsilon_{i}}(-1)^{|S|-1}f_{a}(S),
\end{align}
which can be verified by using the mathematical induction as follows:\\
\indent $1).$ For $h=2$, it can be readily verified that
\[\sum_{\emptyset \not=S\subseteq \{r_{1},r_{2}\}}(-1)^{|S|-1}f_{a}(S)=f_{a}(\{r_{1}\})+f_{a}(\{r_{2}\})-\sum_{\emptyset \not=S\subseteq \Upsilon_{2}}(-1)^{|S|-1}f_{a}(S).\]
\indent $2).$ Suppose that \eqref{equation-thm1-2} holds for $h=s$. Then we have
\begin{align*}
&\sum_{\emptyset \not=S\subseteq \{r_{1},r_{2},\cdots,r_{s+1}\}}(-1)^{|S|-1}f_{a}(S)\\
=&\sum_{\emptyset \not=S\subseteq \{r_{1},r_{2},\cdots,r_{s}\}}(-1)^{|S|-1}f_{a}(S)+f_{a}(\{r_{s+1}\})-\sum_{\emptyset \not=S\subseteq \Upsilon_{s+1}}(-1)^{|S|-1}f_{a}(S)\\
=&\sum_{i=1}^{s+1}f_{a}(\{r_{i}\})-\sum_{i=2}^{s+1}\sum_{\emptyset \not=S\subseteq \Upsilon_{i}}(-1)^{|S|-1}f_{a}(S),
\end{align*}
which implies that \eqref{equation-thm1-2} also holds for $h=s+1$.

Now we claim that $\sum_{\emptyset \not=S\subseteq \Upsilon_{i}}(-1)^{|S|-1}f_{a}(S)\geq0$ for any $2\leq i \leq h$, which implies $|\Omega_1|-\Delta \leq \sum_{i=1}^{h}f_{a}(\{r_{i}\})$.
Indeed, for $a\in \fqm^*$, the Hamming weight of the codeword $\tilde{c}_{a}$ in the linear code $\C_{D_{\Upsilon_{i}}}$ of the form \eqref{CD} with $D_{\Upsilon_{i}}=\bigcup_{r\in \Upsilon_{i}} {\mathbb F}_{q^{r}}$ is \[wt(\tilde{c}_{a})=|D_{\Upsilon_{i}}|-|\{x\in D_{\Upsilon_{i}}: \Tr_{q}^{q^{m}}(ax) = 0 \}|=|D_{\Upsilon_{i}}|-\frac{1}{q}\sum_{x\in D_{\Upsilon_{i}}}\sum_{u\in \fq} \chi(u\Tr_{q}^{q^{m}}(ax)).\]
Again by the principle of inclusion-exclusion, similar to the computation of \eqref{equation-thm1-5}, we have
\[
wt(\tilde{c}_{a})=\sum_{\emptyset \not=S\subseteq \Upsilon_{i}}(-1)^{|S|-1}f_{a}(S)\geq 0.
\]
Therefore, by \eqref{equation-thm1-4}, we can obtain
\[|\Omega_1|-\Delta \leq \sum_{i=1}^{h}f_{a}(\{r_{i}\}) \leq (q-1)\sum_{1\leq i\leq h} q^{r_{i}-1}.\]

On the other hand, for any $S\subseteq \Upsilon$ with $|S|\geq 2$, we have $r_{S}|\gcd(r_{i},r_{j})$ for some $1\leq i<j \leq h$.
Since $\Theta_1$ is nonempty, then for $a\in \Theta_1$, by \eqref{equation-thm1-4} one has that $|\Omega_1|-\Delta=(q-1)\sum_{1\leq i\leq h} q^{r_{i}-1}$ and consequently,
\[d=wt(c_{a})=(q-1)(q^{m-1}-\sum\nolimits_{1\leq i\leq h} q^{r_{i}-1}).\]
Note that $\sum_{1\leq i\leq h} q^{r_{i}}\leq \sum_{0\leq i\leq m-1} q^{i}=\frac{q^m-1}{q-1}\leq q^m-1<q^m$ which implies that $d>0$. This shows that the dimension of $\C_{D}$ is equal to $m$.

A detailed computation using the Griesmer bound gives
\begin{align} \label{equation-thm1-3}
g(m,d)=\sum_{i=0}^{m-1} \lceil \frac{(q-1)(q^{m-1}-\sum_{1\leq i\leq h} q^{r_{i}-1})}{q^i}\rceil
=q^{m}-\sum_{1\leq i\leq h} q^{r_{i}}+h-1
\end{align}
and
\begin{align*}
g(m,d+1)=\sum_{i=0}^{m-1} \lceil \frac{(q-1)(q^{m-1}-\sum_{1\leq i\leq h} q^{r_{i}-1})+1}{q^i}\rceil
=q^{m}-\sum_{1\leq i\leq h} q^{r_{i}}+r_{1}+h-1.
\end{align*}

Thus, the code $\C_{D}$ is distance-optimal if $g(m,d+1)>n$, i.e., $r_{1}+h-1>\sum_{i=1}^{h}q^{r_{i}}-|\Omega_1|$. By Lemmas \ref{self-orthogonalF} and \ref{self-orthogonalD-}, it can be verified that $\C_{D}$ is self-orthogonal if $(q,t)\notin \{(2,1),(2,2),(3,1)\}$.
Note that $|\Omega_1|-\Delta\geq 0$ and $wt(c_{a})\leq (q-1)q^{m-1}$ due to \eqref{equation-thm1-1}. Then, by Lemma \ref{minimal}, $\C_{D}$ is minimal if $q^{m-1}> \sum_{i=1}^{h}q^{r_{i}}$ since $\frac{w_{min}}{w_{max}}\geq\frac{(q-1)(q^{m-1}-\sum_{i=1}^{h}q^{r_{i}-1})}{(q-1)q^{m-1}}>\frac{q-1}{q}$.
 This completes the proof.
\end{proof}

\begin{remark}
Notice that the condition $\Theta_1\ne \emptyset$ can be easily satisfied. For $h=1$, we have
$|\Theta_1|=q^{m}-q^{m-r_{1}}>0$. For $h=2$, we have $|\Theta_1|\geq q^{m-t}-q^{m-r_{1}}-q^{m-r_{2}}>q^{m-t}-q^{m-r_{1}+1}\geq 0$ since $|\{a\in \fqm: \Tr_{q^{r}}^{q^{m}}(a)=0\}|=q^{m-r}$ for any $r|m$ and $t< r_{1}$.
In addition, for a general $h>2$, if $\gcd(r_{i},r_{j})=t$ for any $i<j$, which implies $r_{1}>t$ since $r_{i}\nmid r_{j}$ for any $i<j$, we have
$|\Theta_1| \geq q^{m-t}-\sum_{i=1}^{h}q^{m-r_{i}}\geq q^{m-t}-\sum_{i=t+1}^{m}q^{m-i}=\frac{(q-2)q^{m-t}+1}{q-1}>0$.
\end{remark}

The inequality $r_{1}+h-1>\sum_{i=1}^{h}q^{r_{i}}-|\Omega_1|$ in Theorem \ref{optimalcode-con-1} can be easily satisfied when $r_{1}$ is large enough and $\max_{1\leq i<j \leq h}\{\gcd(r_{i},r_{j})\}$ is small, and then many optimal linear codes can be produced.

In particular, if $h=1$, then $\C_{D}$ in Theorem \ref{optimalcode-con-1} is a Solomon and Stiffler code \cite{GSJS} in the nonprojective case and  its weight distribution can be determined as below.

\begin{thm} \label{optimalcode-con-1-1}
Let $m>1$ and $r<m$ be positive integers with $r|m$. If $D= \fqm \backslash \fqr$, then $\C_{D}$ defined in \eqref{CD} is a $2$-weight $[q^m-q^r,m,(q-1)(q^{m-1}-q^{r-1})]$ linear code with weight enumerator
$$1+(q^{m}-q^{m-r})z^{(q-1)(q^{m-1}-q^{r-1})}+(q^{m-r}-1)z^{(q-1)q^{m-1}}$$
and it is a Griesmer code.
\end{thm}

\begin{proof}
It is obvious that $n=|D|=q^m-q^r$ and  by \eqref{equation-thm1-1} and \eqref{equation-thm1-5}, for $a\in \fqm^{*}$, we have
\[wt(c_{a})=(q-1)q^{m-1}-q^{r}+\Phi({\mathbb F}_{q^{r}}).\]
This together with \eqref{equation-thm1-4} implies that
\begin{align*} 
wt(c_{a})=\left\{\begin{array}{ll}
(q-1)q^{m-1}, & \mbox{if}\,\,\Tr_{q^{r}}^{q^{m}}(a)=0,  \\[0.05in]
(q-1)(q^{m-1}-q^{r-1}),  & \mbox{otherwise}.
\end{array}\right.
\end{align*}
Then the weight distribution of $\C_{D}$ follows from the balanced property of trace functions and
$\C_{D}$ is a Griesmer code due to \eqref{equation-thm1-3}.
This completes the proof.
\end{proof}

\begin{example}
Let $q=3$, $m=6$, $r=2$. Magma experiments show that $\C_{D}$ is a $[720,6,480]$ linear code with the weight enumerator $1+648z^{480}+80z^{486}$, which is consistent with our result in Theorem \ref{optimalcode-con-1-1}. This code is a Griesmer code.
\end{example}

Note that $D=\ftwom \backslash \{0,1\}$ if $q=2$ and $r=1$ in Theorem \ref{optimalcode-con-1-1}. The following result shows that a special class of  Griesmer codes can also be obtained from $D=\fqm \backslash \{0,1\}$ for $q\ne 2$. We omit its proof here since it can be proved in the same manner.

\begin{thm} \label{optimalcode-con-1-3}
Let $q\ne 2 $ and $m>1$ be a positive integer. If $D=\fqm \backslash \{0,1\}$, then $\C_{D}$ defined by \eqref{CD} is a $2$-weight $[q^m-2,m,(q-1)q^{m-1}-1]$ linear code with weight enumerator
$$1+(q^{m-1}-1)z^{(q-1)q^{m-1}}+(q^{m}-q^{m-1})z^{(q-1)q^{m-1}-1}.$$
Moreover, this code is a Griesmer code and it is minimal.
\end{thm}

\begin{example}
Let $q=3$, $m=5$. Magma experiments show that $\C_{D}$ is a $[241,5,161]$ linear code with the weight enumerator $1+162z^{161}+80z^{162}$, which is consistent with our result in Theorem \ref{optimalcode-con-1-3}. This code is a Griesmer code and it is also minimal.
\end{example}

In what follows, we determine the weight distribution of $\C_{D}$ in Theorem \ref{optimalcode-con-1} for $h=2$.
\begin{thm} \label{optimalcode-con-1-2}
Let $\C_{D}$ be defined by \eqref{CD} and \eqref{CD1-D1}. If $h=2$, then $\C_{D}$ is a linear code with parameters $[q^m-q^{r_{2}}-q^{r_{1}}+q^t,m,(q-1)(q^{m-1}-q^{r_{2}-1}-q^{r_{1}-1})]$
and its weight distribution is given by
\begin{center}
\begin{tabular}{lllll}
\hline weight $w$ & Multiplicity $A_{w}$\\ \hline
                                           $0$ & $1$ \\
$(q-1)q^{m-1}$               & $q^{m-r_{2}-r_{1}+t}-1$ \\
$(q-1)(q^{m-1}-q^{r_{1}-1})$ & $ q^{m-r_{2}}-q^{m-r_{2}-r_{1}+t}$ \\
$(q-1)(q^{m-1}-q^{r_{2}-1})$ & $q^{m-r_{1}}-q^{m-r_{2}-r_{1}+t}$ \\
$(q-1)(q^{m-1}-q^{r_{2}-1}-q^{r_{1}-1}+q^{t-1})$ & $q^{m}-q^{m-t}$ \\
$(q-1)(q^{m-1}-q^{r_{2}-1}-q^{r_{1}-1})$ & $q^{m-t}+q^{m-r_{2}-r_{1}+t}-q^{m-r_{2}}-q^{m-r_{1}}$ \\
\hline
\end{tabular}
\end{center}
Moreover, $\C_{D}$ is a near Griesmer code if $(q,t)=(2,1)$ and distance-optimal if $r_{1}+1>q^t$.
\end{thm}

\begin{proof}
Based on the discussions in Theorem \ref{optimalcode-con-1}, one can obtain the parameters of $\C_{D}$ and conclude that the weight $wt(c_{a})$ of $c_{a}\in\C_{D}$ is
\[wt(c_{a})=n-q^{m-1}+\Phi({\mathbb F}_{q^{r_{1}}})+\Phi({\mathbb F}_{q^{r_{2}}})-\Phi({\mathbb F}_{q^{t}})\]
for $a\in \fqm^{*}$. Then, by \eqref{equation-thm1-4}, for $a\in \fqm^{*}$, one gets
\begin{align*} 
wt(c_{a})=\left\{\begin{array}{lllll}
w_1:=(q-1)q^{m-1}, &
\mbox{if}\,\,\Tr_{q^{r_{2}}}^{q^{m}}(a)=0,\,\,\Tr_{q^{r_{1}}}^{q^{m}}(a)=0 ,  \\
w_2:=(q-1)(q^{m-1}-q^{r_{1}-1}), &
\mbox{if}\,\,\Tr_{q^{r_{2}}}^{q^{m}}(a)=0,\,\,\,\Tr_{q^{r_{1}}}^{q^{m}}(a)\not=0,  \\
w_3:=(q-1)(q^{m-1}-q^{r_{2}-1}), &
\mbox{if}\,\,\Tr_{q^{r_{2}}}^{q^{m}}(a)\not=0,\,\,\Tr_{q^{r_{1}}}^{q^{m}}(a)=0,  \\
w_4:=(q-1)(q^{m-1}-q^{r_{2}-1}-q^{r_{1}-1}+q^{t-1}),  &
\mbox{if}\,\,\Tr_{q^{r_{2}}}^{q^{m}}(a)\Tr_{q^{r_{1}}}^{q^{m}}(a)\Tr_{q^{t}}^{q^{m}}(a)\not=0 ,  \\
w_5:=(q-1)(q^{m-1}-q^{r_{2}-1}-q^{r_{1}-1}), &
\mbox{otherwise}
\end{array}\right.
\end{align*}
which together with Lemma \ref{lem-trace} indicates that
\begin{align*}
A_{w_{1}}+1=&|\{x\in \fqm: \Tr_{q^{r_{2}}}^{q^{m}}(x)=\Tr_{q^{r_{1}}}^{q^{m}}(x)=0\}|
=\frac{1}{q^{r_{1}}}|\{x\in \fqm: \Tr_{q^{r_{2}}}^{q^{m}}(x^{q^{r_{1}}}-x)=0\}|\\
=&\frac{1}{q^{r_{2}+r_{1}}}\sum_{x\in \fqm}\sum_{y\in {\mathbb F}_{q^{r_{2}}}}\chi(\Tr_{q}^{q^{r_{2}}}(y\Tr_{q^{r_{2}}}^{q^{m}}(x^{q^{r_{1}}}-x)))\\
=&\frac{1}{q^{r_{2}+r_{1}}}\sum_{y\in {\mathbb F}_{q^{r_{2}}}}\sum_{x\in \fqm}\chi(\Tr_{q}^{q^{m}}(y(x^{q^{r_{1}}}-x)))\\
=&\frac{1}{q^{r_{2}+r_{1}}}\sum_{y\in {\mathbb F}_{q^{r_{2}}}}\sum_{x\in \fqm}\chi(\Tr_{q}^{q^{m}}((y-y^{q^{r_{1}}})x))\\
=&q^{m-r_{2}-r_{1}+t}.
\end{align*}
The last equal sign holds since the equation $y^{q^{r_{1}}}=y$ has exactly $q^t$ solutions in ${\mathbb F}_{q^{r_{2}}}$.
Further, we have $A_{w_{2}}=q^{m-r_{2}}-q^{m-r_{2}-r_{1}+t}$ and $A_{w_{3}}=q^{m-r_{1}}-q^{m-r_{2}-r_{1}+t}$ since $A_{w_{1}}+A_{w_{2}}= |\{a\in \fqm^{*}: \Tr_{q^{r_{2}}}^{q^{m}}(a)=0\}|=q^{m-r_{2}}-1$ and $A_{w_{1}}+A_{w_{3}}= |\{a\in \fqm^{*}: \Tr_{q^{r_{1}}}^{q^{m}}(a)=0\}|=q^{m-r_{1}}-1$.

Since $0 \notin D$, the minimum distance of $\C_{D}^{\perp}$ is bigger than one, i.e., $A_{1}^{\perp}=0$. From the Pless Power Moments (see \cite{HWPV}, page 259), we have
$$\left\{\begin{array}{lll}
A_{w_{1}}+A_{w_{2}}+A_{w_{3}}+A_{w_{4}}+A_{w_{5}}=q^{m}-1,\\
w_{1}A_{w_{1}}+w_{2}A_{w_{2}}+w_{3}A_{w_{3}}+w_{4}A_{w_{4}}+w_{5}A_{w_{5}}=q^{m-1}(q-1)n.
\end{array}\right.$$
Solving the above equations gives the values of $A_{w_{4}}$ and $A_{w_{5}}$. On the other hand, using \eqref{equation-thm1-3}, we have $g(m,d)=q^m-q^{r_{2}}-q^{r_{1}}+1$ and then $\C_{D}$ is a near Griesmer code if $(q,t)=(2,1)$. The code $\C_{D}$ is distance-optimal if $r_{1}+1>q^t$
according to Theorem \ref{optimalcode-con-1}.
This completes the proof.
\end{proof}

\begin{example}
Let $q=2$, $m=6$, $r_{2}=3$ and $r_{1}=2$. Magma experiments show that $\C_{D}$ is a $[54,6,26]$ linear code with the weight enumerator $1+12z^{26}+32z^{27}+12z^{28}+4z^{30}+3z^{32}$, which is consistent with our result in Theorem \ref{optimalcode-con-1-2}. This code is a near Griesmer code and it is distance-optimal due to \cite{GMB}.
\end{example}

\section{The second family of optimal linear codes} \label{section-4}

In this section, we investigate the linear codes $\C_{D}$ of the form \eqref{CD} for
\begin{eqnarray}\label{CD1-D3}
D= \fqm \backslash \Omega_2,\;\;\Omega_2=\cup_{i=0}^{h}(\theta_{i}+\fqr),
\end{eqnarray}
where $m>1$, $r<m$ are positive integers satisfying $r|m$ and $\theta_{0}=0$, $\theta_{i}\in \fqm^*$ for $1\leq i\leq h$ satisfying $\theta_{i}-\theta_{j}\notin\fqr$ for any $0\leq i<j \leq h$.

For simplicity, define
\begin{eqnarray*}
\Theta_2=\{a\in \fqm: \Tr_{q^{r}}^{q^{m}}(a)=0\,\,{\rm and}\,\,\Tr_{q}^{q^{m}}(a\theta_{i})\not=0\,\,{\rm for\,\,any\,\,} 1\leq i \leq h\}.
\end{eqnarray*}

\begin{thm} \label{optimalcode-con-3}
Let $\C_{D}$ be defined by \eqref{CD} and \eqref{CD1-D3}, where $h$ is a positive integer satisfying $h+1<q^{m-r}$ if $h+1\leq q$ and otherwise $h<(q-1)q^{m-r-1}$ and $\Theta_2$ is nonempty.
Then \\
$1).$ $\C_{D}$ is a $[q^m-(h+1)q^r,m]$ linear code with minimal distance $d=(q-1)(q^{m-1}-(h+1)q^{r-1})$ (resp. $d=(q-1)q^{m-1}-hq^{r}$) if $h+1\leq q$ (resp. $h+1>q$);\\
$2).$ $\C_{D}$ is a Griesmer code if $h+1\leq q$; and when $h+1>q$, it is distance-optimal if
$(h+1)q^r+r>1+hq\frac{q^r-1}{q-1}+\sum_{i=r}^{m-1}\lfloor \frac{hq^r-1}{q^i}\rfloor$;\\
$3).$ $\C_{D}$ is at most $(h+2)$-weight and its weights take values from
\[\{(q-1)(q^{m-1}-(h+1)q^{r-1})\}\cup\{(q-1)q^{m-1}-iq^{r}:i=0,1,2,\cdots, h\};\]
$4).$ $\C_{D}$ is self-orthogonal if $(q,r)\notin \{(2,1),(2,2),(3,1)\}$;\\
$5).$ $\C_{D}$ is minimal if $h+1<q^{m-r-1}$ (resp. $h<(q-1)q^{m-r-2}$) when $h+1\leq q$ (resp. $h+1> q$).
\end{thm}

\begin{proof}
Observe that the length of $\C_{D}$ is $n=|D|=q^m-(h+1)q^r$ since $(\theta_{i}+\fqr)\cap(\theta_{j}+\fqr)=\emptyset$ due to $\theta_{i}-\theta_{j}\notin\fqr$ for any $0\leq i<j \leq h$. For $a\in \fqm^*$, the Hamming weight $wt(c_{a})$ of the codeword $c_{a}$ in $\C_{D}$ is $n-N_{a}$, where $N_{a}=|\{x\in \fqm\backslash \Omega_2: \Tr_{q}^{q^{m}}(ax) = 0\}|$. Using the orthogonal property of nontrivial additive characters gives
\begin{align*}
N_{a}=&\frac{1}{q}\sum_{x\in \fqm\backslash \Omega_2}\sum_{u\in \fq} \chi(u\Tr_{q}^{q^{m}}(ax))
=\frac{1}{q}\sum_{u\in \fq}(\sum_{x\in \fqm} \chi(u\Tr_{q}^{q^{m}}(ax))
-\sum_{i=0}^{h}\sum_{x\in (\theta_{i}+\fqr)} \chi(u\Tr_{q}^{q^{m}}(ax)))\\
=&q^{m-1}-\frac{1}{q}\sum_{u\in \fq}\sum_{i=0}^{h}\sum_{x\in \fqr} \chi(u\Tr_{q}^{q^{m}}(a(x+\theta_{i}))).
\end{align*}
Hence, for $a\in \fqm^*$, we have
\[wt(c_{a})=(q-1)q^{m-1}-(h+1)q^r+\frac{1}{q}\sum_{i=0}^{h}\sum_{u\in \fq}\sum_{x\in \fqr} \chi(u\Tr_{q}^{q^{m}}(a(x+\theta_{i}))).\]
This together with the following fact
\begin{align*}
\sum_{u\in \fq}\sum_{x\in \fqr} \chi(u\Tr_{q}^{q^{m}}(a(x+\theta)))=\left\{\begin{array}{ll}
q^{r+1}, & \mbox{if $\Tr_{q^{r}}^{q^{m}}(a)=0$ and $\Tr_{q}^{q^{m}}(a\theta)=0$},  \\
0,  & \mbox{if $\Tr_{q^{r}}^{q^{m}}(a)=0$ and $\Tr_{q}^{q^{m}}(a\theta)\not=0$},  \\
q^{r},  & \mbox{if $\Tr_{q^{r}}^{q^{m}}(a)\not=0$},
\end{array}\right.
\end{align*}
where $\theta\in \fqm$, one can claim that $wt(c_{a})$ takes values from
\[\{(q-1)(q^{m-1}-(h+1)q^{r-1})\}\cup\{(q-1)q^{m-1}-iq^{r}:i=0,1,2,\cdots, h\}.\]
Moreover, $wt(c_{a})=(q-1)(q^{m-1}-(h+1)q^{r-1})$ if and only if $\Tr_{q^{r}}^{q^{m}}(a)\not=0$ and $wt(c_{a})=(q-1)q^{m-1}-hq^{r}$ if and only if $\Tr_{q^{r}}^{q^{m}}(a)=0$ and $\Tr_{q}^{q^{m}}(a\theta_{i})\not=0$ for any $1\leq i \leq h$.

Let $w_1=(q-1)(q^{m-1}-(h+1)q^{r-1})$ and $w_2=(q-1)q^{m-1}-hq^{r}$.  Then the above discussion indicates that $A_{w_1}=|\{a\in \fqm: \Tr_{q^{r}}^{q^{m}}(a)\not=0\}|=q^{m}-q^{m-r}>0$ and $A_{w_2}>0$ if $\Theta_2$ is nonempty. If $h+1\leq q$, then we have $w_1\leq w_2$ which means that $d=w_1>0$ due to $q^{m-r}>h+1$. If $h+1> q$, then we have $w_1>w_2$ and consequently, $d=w_2>0$ since $\Theta_2 \ne \emptyset$ and $h<(q-1)q^{m-r-1}$. This also shows that the dimension of $\C_{D}$ is equal to $m$.

According to the Griesmer bound, for $h+1\leq q$,  one obtains
\begin{align*}
g(m,d)=\sum_{i=0}^{m-1} \lceil \frac{(q-1)(q^{m-1}-(h+1)q^{r-1})}{q^i}\rceil
=q^{m}-(h+1)q^{r},
\end{align*}
which implies that $\C_{D}$ is a Griesmer code.
For the case $h+1> q$, we have
\begin{align*}
g(m,d)=\sum_{i=0}^{m-1} \lceil \frac{(q-1)q^{m-1}-hq^{r}}{q^i}\rceil
=q^m-1-hq\frac{q^r-1}{q-1}-\sum_{i=0}^{m-r-1}\lfloor \frac{h}{q^i}\rfloor
\end{align*}
and
\begin{align*}
g(m,d+1)=\sum_{i=0}^{m-1} \lceil \frac{(q-1)q^{m-1}-hq^{r}+1}{q^i}\rceil
=q^m+r-1-hq\frac{q^r-1}{q-1}-\sum_{i=r}^{m-1}\lfloor \frac{hq^r-1}{q^i}\rfloor.
\end{align*}
Note that $n>g(m,d)$ when $h+1>q$ since $n-g(m,d)=hq\frac{q^r-1}{q-1}+\sum_{i=0}^{m-r-1}\lfloor \frac{h}{q^i}\rfloor-(h+1)q^r+1=h(q+\cdots+q^{r-1})-q^r+\sum_{i=0}^{m-r-1}\lfloor \frac{h}{q^i}\rfloor+1>0$.
Thus, when $h+1> q$, $\C_{D}$ is distance-optimal if
$(h+1)q^r+r>1+hq\frac{q^r-1}{q-1}+\sum_{i=r}^{m-1}\lfloor \frac{hq^r-1}{q^i}\rfloor$.

The self-orthogonality and minimality of $\C_{D}$ can be readily verified by Lemmas \ref{self-orthogonalF}, \ref{self-orthogonalD-} and \ref{minimal}. This completes the proof.
\end{proof}

\begin{remark}
The code in Theorem \ref{optimalcode-con-3} is a Griesmer code if $h+1\leq q$. When $h=0$, it has the same parameters and weight distribution as the code in Theorem \ref{optimalcode-con-1-1}, and it has different parameters with the Solomon and Stiffler codes in the nonprojective case when $1\leq h\leq q-1$.
\end{remark}

\begin{remark} \label{remark-7}
The condition $\Theta_2\ne \emptyset$ for $h+1>q$ can be easily satisfied.  For example, let $a\in\fqm^{*}$ with $\Tr_{q^{r}}^{q^{m}}(a)=0$ and $\theta_{1},\theta_{2},\cdots,\theta_{h} \in \Lambda$ satisfy $\theta_{i}\notin \theta_{j}+\fqr$ for any $1\leq i<j \leq h$, where $\Lambda=\{\theta\in \fqm:\Tr_{q}^{q^{m}}(a\theta)\not=0\}$. If $(q-1)q^{m-1}-hq^{r}\geq 0 $, namely, $h\leq (q-1)q^{m-r-1}$, then there must exist $\theta_{i}$'s such that $\Theta_2\ne \emptyset$ since $|\Lambda|=(q-1)q^{m-1}$.
Specially, for $h=2$, which implies $q=2$, $|\Theta_2|=\tau q^{m-r-2}>0$ always holds due to $\tau\geq 1$ (see Theorem \ref{optimalcode-con-3-2} below).
\end{remark}

\begin{example}
Let $q=4$, $m=6$, $r=2$, $\theta_{1}=\alpha$, $\theta_{2}=\alpha^2$ and $\theta_{3}=\alpha^3$, where $\alpha$ is a primitive element of $\fqm$. Magma experiments show that $\C_{D}$ is a $[4032,6,3024]$ linear code with the weight enumerator
$1+3948z^{3024}+108z^{3040}+36z^{3056}+3z^{3072}$, which is consistent with our result in Theorem \ref{optimalcode-con-3}. This code is a Griesmer code and it is also self-orthogonal and minimal.
\end{example}

In what follows, we determine the weight distributions of $\C_{D}$ in Theorem \ref{optimalcode-con-3} for $h=1$ and $h=2$ respectively.

\begin{thm} \label{optimalcode-con-3-1}
Let $\C_{D}$ be defined by \eqref{CD} and \eqref{CD1-D3} with $h=1$ and $q^{m-r}>2$.
Then $\C_{D}$ is a $[q^m-2q^r,m,(q-1)(q^{m-1}-2q^{r-1})]$ Griesmer code with the following weight distribution:
\begin{center}
\begin{tabular}{lll}
\hline weight $w$ & Multiplicity $A_{w}$\\ \hline
                                           $0$ & $1$ \\
$(q-1)q^{m-1}$               & $q^{m-r-1}-1$ \\
$(q-1)q^{m-1}-q^{r}$ & $(q-1)q^{m-r-1}$ \\
$(q-1)(q^{m-1}-2q^{r-1})$ & $q^{m-r}(q^{r}-1)$ \\
\hline
\end{tabular}
\end{center}
\end{thm}

\begin{proof}
According to the proof of Theorem \ref{optimalcode-con-3}, one can get the parameters of $\C_{D}$ and conclude that the nonzero weights of $\C_{D}$ take values from $\{w_1:=(q-1)(q^{m-1}-2q^{r-1}), w_2:=(q-1)q^{m-1}-q^{r}, w_3:=(q-1)q^{m-1} \}$ with $A_{w_1}=q^m-q^{m-r}$ if $h=1$. Using the Pless Power Moments (see \cite{HWPV}, page 259) and the fact $A_{1}^{\perp}=0$ since  $0 \notin D$ give
$$\left\{\begin{array}{lll}
A_{w_{1}}+A_{w_{2}}+A_{w_{3}}=q^{m}-1,\\
w_{1}A_{w_{1}}+w_{2}A_{w_{2}}+w_{3}A_{w_{3}}=q^{m-1}(q-1)n,
\end{array}\right.$$
which leads to $A_{w_2}=(q-1)q^{m-r-1}$ and $A_{w_3}=q^{m-r-1}-1$.  The code $\C_{D}$ is a Griesmer code follows from  Theorem \ref{optimalcode-con-3}. This completes the proof.
\end{proof}

\begin{example}
Let $q=2$, $m=6$, $r=2$ and $\theta_{1}=\alpha$, where $\alpha$ is a primitive element of $\fqm$. Magma experiments show that $\C_{D}$ is a $[56,6,28]$ binary linear code with the weight enumerator $1+56z^{28}+7z^{32}$, which is consistent with our result in Theorem \ref{optimalcode-con-3-1}. This code is a Griesmer code and distance-optimal due to \cite{GMB}.
\end{example}

\begin{example}
Let $q=3$, $m=4$, $r=2$ and $\theta_{1}=\alpha$, where $\alpha$ is a primitive element of $\fqm$. Magma experiments show that $\C_{D}$ is a $[63,4,42]$ linear code with the weight enumerator $1+72z^{42}+6z^{45}+2z^{54}$, which is consistent with our result in Theorem \ref{optimalcode-con-3-1}. This code is a Griesmer code and distance-optimal due to \cite{GMB}.
\end{example}

\begin{thm} \label{optimalcode-con-3-2}
Let $\C_{D}$ be defined by \eqref{CD} and \eqref{CD1-D3} with $h=2$. Let $\tau=|\{(u,v)\in \fq^2: u\theta_{1}+v\theta_{2}\in \fqr \}|$ and $m, r$ be positive integers satisfying $m>r+2$ when $q=2$ and $m>r+1$ when $q=3$. Then the weight distribution of  $\C_{D}$ is given by
\begin{center}
\begin{tabular}{lll}
\hline weight $w$ & Multiplicity $A_{w}$\\ \hline
                                           $0$ & $1$ \\
$(q-1)(q^{m-1}-3q^{r-1})$               & $q^{m-r}(q^{r}-1)$ \\
$(q-1)q^{m-1}$ & $\tau q^{m-r-2}-1$ \\
$(q-1)q^{m-1}-2q^{r}$ & $(q^{2}-2q+\tau)q^{m-r-2}$ \\
$(q-1)q^{m-1}-q^{r}$ & $2(q-\tau)q^{m-r-2}$ \\
\hline
\end{tabular}
\end{center}
Moreover, $\C_{D}$ is a Griesmer code when $q>2$ and for $q=2$ it is distance-optimal if $r=1$.
\end{thm}

\begin{proof}
For $h=2$, by Theorem \ref{optimalcode-con-3}, the nonzero weight $wt(c_{a})$ of $c_{a}\in\C_{D}$ belongs to $\{w_{1}:=(q-1)(q^{m-1}-3q^{r-1}), w_{2}:=(q-1)q^{m-1}, w_{3}:=(q-1)q^{m-1}-2q^{r}, w_{4}:=(q-1)q^{m-1}-q^{r}\}$. Moreover, similar to the proof of Theorem \ref{optimalcode-con-3}, for $a\in \fqm^*$, one can obtain
\begin{align*}
wt(c_{a})=\left\{\begin{array}{ll}
(q-1)(q^{m-1}-3q^{r-1}), & \mbox{if $\Tr_{q^{r}}^{q^{m}}(a)\not=0$},  \\
(q-1)q^{m-1},& \mbox{if $\Tr_{q^{r}}^{q^{m}}(a)=0$ and $\Tr_{q}^{q^{m}}(a\theta_{1})=0$ and $\Tr_{q}^{q^{m}}(a\theta_{2})=0$}, \\
(q-1)q^{m-1}-2q^{r},& \mbox{if $\Tr_{q^{r}}^{q^{m}}(a)=0$ and $\Tr_{q}^{q^{m}}(a\theta_{1})\not=0$ and $\Tr_{q}^{q^{m}}(a\theta_{2})\not=0$}, \\
(q-1)q^{m-1}-q^{r},   & \mbox{otherwise}
\end{array}\right.
\end{align*}
and consequently one gets $A_{w_{1}}=q^m-q^{m-r}$ and
\begin{align*}
1+A_{w_{2}}=&|\{x\in \fqm: \Tr_{q^{r}}^{q^{m}}(x)=\Tr_{q}^{q^{m}}(\theta_{1}x)=\Tr_{q}^{q^{m}}(\theta_{2}x)=0\}|\\
=&\frac{1}{q^r}|\{x\in \fqm: \Tr_{q}^{q^{m}}(\theta_{1}(x^{q^r}-x))=\Tr_{q}^{q^{m}}(\theta_{2}(x^{q^r}-x))=0\}|\\
=&\frac{1}{q^{r+2}}\sum_{x\in \fqm}\sum_{u\in \fq}\chi(u\Tr_{q}^{q^{m}}(\theta_{1}(x^{q^r}-x)))
\sum_{v\in \fq}\chi(v\Tr_{q}^{q^{m}}(\theta_{2}(x^{q^r}-x)))\\
=&\frac{1}{q^{r+2}}\sum_{u\in \fq}\sum_{v\in \fq}\sum_{x\in \fqm}
\chi(\Tr_{q}^{q^{m}}((u\theta_{1}+v\theta_{2})(x^{q^r}-x)))\\
=&\frac{1}{q^{r+2}}\sum_{u\in \fq}\sum_{v\in \fq}\sum_{x\in \fqm}
\chi(\Tr_{q}^{q^{m}}(((u\theta_{1}+v\theta_{2})-(u\theta_{1}+v\theta_{2})^{q^r})x))\\
=&\tau q^{m-r-2}
\end{align*}
by using Lemma \ref{lem-trace}.
Then, the weight distribution of $\C_{D}$ follows from the first two Pless Power Moments as we did before.

Theorem \ref{optimalcode-con-3} implies that $\C_{D}$ is a Griesmer code when $q\geq3$. For $q=2$, we have
\[g(m,d+1)=\sum_{i=0}^{m-1} \lceil \frac{2^{m-1}-2^{r+1}+1}{2^i}\rceil
=2^m-2^{r+2}+r+2,
\]
which implies that $\C_{D}$ for $q=2$ is distance-optimal if $r=1$ since $g(m,d+1)-n=r+2-2^r>0$ holds if and only if $r=1$.
This completes the proof.
\end{proof}

\begin{remark}
Note that $\tau\geq 1$ since $(0,0)\in \{(u,v)\in \fq^2: u\theta_{1}+v\theta_{2}\in \fqr \}$ and $\C_{D}$ in Theorem \ref{optimalcode-con-3-2} is reduced to $3$-weight if $q=3$. In particular, we have $\tau=1$ if $q=2$ by the definition of $\theta_{i}$'s.
\end{remark}

\begin{example}
Let $q=2$, $m=6$, $r=1$, $\theta_{1}=\alpha$ and $\theta_{2}=\alpha^2$, where $\alpha$ is a primitive element of $\fqm$. Magma experiments show that $\C_{D}$ is a $[58,6,28]$ binary linear code with the weight enumerator $1+8z^{28}+32z^{29}+16z^{30}+7z^{32}$, which is consistent with our result in Theorem \ref{optimalcode-con-3-2}. This code is distance-optimal due to \cite{GMB}.
\end{example}

\begin{example}
Let $q=3$, $m=4$, $r=1$, $\theta_{1}=\alpha$ and $\theta_{2}=\alpha^2$, where $\alpha$ is a primitive element of $\fqm$. Magma experiments show that $\C_{D}$ is a $[72,4,48]$ linear code with the weight enumerator $1+66z^{48}+12z^{51}+2z^{54}$, which is consistent with our result in Theorem \ref{optimalcode-con-3-2}. This code is a Griesmer code and it is distance-optimal due to \cite{GMB}.
\end{example}

\section{The third family of optimal linear codes} \label{section-5}
In this section, we study the linear codes $\C_{D}$ of the form \eqref{CD} for
\begin{eqnarray}\label{CD1-D4}
D= \fqm \backslash \Omega_3,\;\;\Omega_3=\cup_{i=1}^{h}(\theta_{i}*\fqr),
\end{eqnarray}
where $m>1$, $r<m$ are positive integers satisfying $r|m$ and $\theta_{i}\in \fqm^{*}$ for $1\leq i \leq h$ satisfying $\theta_{i}/\theta_{j}\notin\fqr$ for any $1\leq i<j \leq h$. For simplicity, define
\begin{eqnarray*}
\Theta_3=\{a\in \fqm:\Tr_{q^{r}}^{q^{m}}(a\theta_{i})\not=0\,\,{\rm for\,\,any\,\,} 1\leq i \leq h\}.
\end{eqnarray*}

\begin{thm} \label{optimalcode-con-4}
Let $\C_{D}$ be defined by \eqref{CD} and \eqref{CD1-D4}. If $h<q^{m-r}$ and $\Theta_3\ne \emptyset$, then\\
$1).$ $\C_{D}$ is a $[q^m-hq^r+h-1,m,(q-1)(q^{m-1}-hq^{r-1})]$ linear code;\\
$2).$ $\C_{D}$ is at most $(h+1)$-weight and its weights take values from
\[\{(q-1)(q^{m-1}-iq^{r-1}): i=0,1,2,\cdots,h\};\]
$3).$ $\C_{D}$ is a Griesmer code if and only if $h=1$ and it is a near Griesmer code if and only if $h=2$ or $(q,h)=(2,3)$. When $h>1$, $\C_{D}$ is distance-optimal if $r>\sum_{i=r}^{m-1}\lfloor \frac{h(q-1)q^{r-1}-1}{q^i}\rfloor$;\\
$4).$ $\C_{D}$ is self-orthogonal if $(q,r)\notin \{(2,1),(2,2),(3,1)\}$;\\
$5).$ $\C_{D}$ is minimal if $q^{m-r-1}> h$.
\end{thm}

\begin{proof}
The length of $\C_{D}$ is $n=|D|=q^m-hq^r+h-1$ since $(\theta_{i}*\fqr)\cap(\theta_{j}*\fqr)=\{0\}$ due to $\theta_{i}/\theta_{j}\notin\fqr$ for any $1\leq i<j \leq h$. For $a\in \fqm^*$, the Hamming weight $wt(c_{a})$ of the codeword $c_{a}$ in $\C_{D}$ is $n-N_{a}$, where $N_{a}=|\{x\in \fqm\backslash \Omega_3: \Tr_{q}^{q^{m}}(ax) = 0\}|$. Using the orthogonal property of nontrivial additive characters leads to
\begin{align*}
N_{a}=&\frac{1}{q}\sum_{x\in \fqm\backslash \Omega_3}\sum_{u\in \fq} \chi(u\Tr_{q}^{q^{m}}(ax))\\
=&\frac{1}{q}\sum_{u\in \fq}(\sum_{x\in \fqm} \chi(u\Tr_{q}^{q^{m}}(ax))
-\sum_{i=1}^{h}\sum_{x\in \theta_{i}*\fqr} \chi(u\Tr_{q}^{q^{m}}(ax)))+h-1\\
=&\frac{1}{q}(q^m-\sum_{i=1}^{h}\sum_{u\in \fq}\sum_{x\in \fqr} \chi(u\Tr_{q}^{q^{m}}(a\theta_{i}x)))+h-1,
\end{align*}
which leads to
\[wt(c_{a})=(q-1)q^{m-1}-hq^r+\frac{1}{q}\sum_{i=1}^{h}\sum_{u\in \fq}\sum_{x\in \fqr} \chi(u\Tr_{q}^{q^{m}}(a\theta_{i}x)).\]
Note that
\begin{align*}
\frac{1}{q}\sum_{u\in \fq}\sum_{x\in \fqr} \chi(u\Tr_{q}^{q^{m}}(a\theta x))=\left\{\begin{array}{ll}
q^{r},    & \mbox{if $\Tr_{q^{r}}^{q^{m}}(a\theta)=0$},  \\
q^{r-1},  & \mbox{if $\Tr_{q^{r}}^{q^{m}}(a\theta)\ne0$}
\end{array}\right.
\end{align*}
holds for any $\theta\in \fqm^{*}$. Thus, for $a\in\fqm^*$, one can conclude tha $wt(c_{a})$ takes value from
\[\left\{(q-1)q^{m-1},(q-1)(q^{m-1}-q^{r-1}),\cdots,(q-1)(q^{m-1}-hq^{r-1})\right\}\]
and $wt(c_{a})=(q-1)(q^{m-1}-hq^{r-1})$ if and only if $\Tr_{q^{r}}^{q^{m}}(a\theta_{i})\not=0$ for any $1\leq i \leq h$.
Thus, $\C_{D}$ is at most $(h+1)$-weight.

Since $\Theta_{3}\ne \emptyset$, we have $d=(q-1)(q^{m-1}-hq^{r-1})>0$ due to $q^{m-r}>h$. This implies that the dimension of $\C_{D}$ is $m$. According to the Griesmer bound,   we have
\begin{align*}
g(m,d)=\sum_{i=0}^{m-1} \lceil \frac{(q-1)(q^{m-1}-hq^{r-1})}{q^i}\rceil
=q^{m}-hq^{r}+h-1-\sum_{i=1}^{m-r}\lfloor \frac{h(q-1)}{q^i}\rfloor.
\end{align*}
It can be readily verified that $n-g(m,d)=0$ if and only if $h=1$ and $n-g(m,d)=1$ if and only if $h=2$ or $(q,h)=(2,3)$.
When $h>1$, we have
\begin{align*}
g(m,d+1)=\sum_{i=0}^{m-1} \lceil \frac{(q-1)(q^{m-1}-hq^{r-1})+1}{q^i}\rceil
=n+r-\sum_{i=r}^{m-1}\lfloor \frac{h(q-1)q^{r-1}-1}{q^i}\rfloor.
\end{align*}
Therefore, $\C_{D}$ is distance-optimal if $r>\sum_{i=r}^{m-1}\lfloor \frac{h(q-1)q^{r-1}-1}{q^i}\rfloor$.

Then the proof is completed due to  Lemmas \ref{self-orthogonalF}, \ref{self-orthogonalD-} and \ref{minimal}.
\end{proof}

\begin{remark}
The condition $\Theta_3\ne \emptyset$ always holds for $h< q^{r}$ since $|\Theta_3|\geq q^m-hq^{m-r}>0$ due to the fact that $|\{a\in \fqm: \Tr_{q^{r}}^{q^{m}}(a\theta_{i})=0\}|=q^{m-r}$ for any $\theta_{i}\ne 0$. Moreover, similar to the discussion in Remark \ref{remark-7}, we conclude that there must exist $\theta_{i}$'s such that $|\Theta_3|>0$ for $h\leq q^{m-r}$ since $\theta_{i}/\theta_{j}\notin\fqr^*$ for any $1\leq i<j \leq h$. In addition, assume that $\theta_{i}$'s are linearly independent over $\fqr$ for $1\leq i \leq h$ which implies $h\leq m/r$, then by the property of trace functions, for any $(v_{1},\cdots,v_{h})\in \fqr^{h}$, we have $|\{x\in \fqm: \Tr_{q^{r}}^{q^{m}}(x\theta_{i})=-v_{i}\,\,{\rm for\,\,all\,\,} 1\leq i \leq h\}|=q^{m-hr}$ which indicates that $|\Theta_3|=(q^{r}-1)^{h}q^{m-hr}>0$.
\end{remark}

\begin{example}
Let $q=2$, $m=12$, $r=3$, $\theta_{1}=1$, $\theta_{2}=\alpha$, $\theta_{3}=\alpha^2$ and $\theta_{4}=\alpha^3$, where $\alpha$ is a primitive element of $\fqm$. Magma experiments show that $\C_{D}$ is a $[4067,12,2032]$ binary linear code with the weight enumerator $1+2401z^{2032}+1372z^{2036}+294z^{2040}+28z^{2044}$, which is consistent with our result in Theorem \ref{optimalcode-con-4}. This code is distance-optimal due to the Griesmer bound and it is also self-orthogonal and minimal.
\end{example}

The Griesmer code $\C_{D}$ in Theorem \ref{optimalcode-con-4} for $h=1$ has the same parameters and weight distribution as the one in Theorem \ref{optimalcode-con-1-1}. In the following, we determine the weight distribution of $\C_{D}$ in Theorem \ref{optimalcode-con-4} for $h=2$.

\begin{thm} \label{optimalcode-con-4-1}
Let $\C_{D}$ be defined by \eqref{CD} and \eqref{CD1-D4} with $h=2$ and $q^{m-r}>2$. Then $\C_{D}$ is a $3$-weight $[q^m-2q^r+1,m,(q-1)(q^{m-1}-2q^{r-1})]$ linear code with the following weight distribution:
\begin{center}
\begin{tabular}{lll}
\hline weight $w$ & Multiplicity $A_{w}$\\ \hline
                                           $0$ & $1$ \\
$(q-1)q^{m-1}$               & $q^{m-2r}-1$ \\
$(q-1)(q^{m-1}-2q^{r-1})$ & $q^m-2q^{m-r}+q^{m-2r}$ \\
$(q-1)(q^{m-1}-q^{r-1})$ & $2(q^{m-r}-q^{m-2r})$ \\
\hline
\end{tabular}
\end{center}
Moreover, $\C_{D}$ is a near Griesmer code and it is distance-optimal if
$r+\lfloor \frac{2}{q}\rfloor>1$.
\end{thm}

\begin{proof}
According to the proof of Theorem \ref{optimalcode-con-4}, one can get the parameters of $C_{D}$ and claim that the weight of $c_{a}\in C_{D}$ for $a\in\fqm^*$ is
\begin{align}\label{thm10-equ-1}
wt(c_{a})=\left\{\begin{array}{ll}
w_{1}:=(q-1)q^{m-1}, & \mbox{if $\Tr_{q^{r}}^{q^{m}}(a\theta_{1})=0$ and $\Tr_{q^{r}}^{q^{m}}(a\theta_{2})=0$},  \\
w_{2}:=(q-1)(q^{m-1}-2q^{r-1}), & \mbox{if $\Tr_{q^{r}}^{q^{m}}(a\theta_{1})\ne0$ and $\Tr_{q^{r}}^{q^{m}}(a\theta_{2})\ne0$},  \\
w_{3}:=(q-1)(q^{m-1}-q^{r-1}),  & \mbox{otherwise}.
\end{array}\right.
\end{align}
By Lemma \ref{lem-trace}, $\theta_{2}/\theta_{1}\notin \fqr$ and \eqref{thm10-equ-1}, one obtains
\begin{align*}
1+A_{w_{1}}
=&|\{x\in \fqm: \Tr_{q^{r}}^{q^{m}}(x)=0\,\,
{\rm and}\,\,\Tr_{q^{r}}^{q^{m}}(x\theta_{2}/\theta_{1})=0\}| \\
=&\frac{1}{q^r}|\{x\in \fqm: \Tr_{q^{r}}^{q^{m}}(\theta_{2}/\theta_{1}(x^{q^r}-x))=0\}|\\
=&\frac{1}{q^r}|\{x\in \fqm: \Tr_{q^{r}}^{q^{m}}((\theta_{2}/\theta_{1}-(\theta_{2}/\theta_{1})^{q^r})x)=0\}|\\
=&q^{m-2r}.
\end{align*}
Then, the weight distribution of $\C_{D}$ follows from the first two Pless Power Moments.
$\C_{D}$ is a near Griesmer code due to Theorem \ref{optimalcode-con-4}. By the Griesmer bound, we have
\begin{align*}
g(m,d+1)=\sum_{i=0}^{m-1} \lceil \frac{(q-1)(q^{m-1}-2q^{r-1})+1}{q^i}\rceil
=\left\{\begin{array}{ll}
q^{m}-2q^{r}+r+1, & \mbox{if $q=2$},  \\
q^{m}-2q^{r}+r, & \mbox{if $q>2$}  \\
\end{array}\right.
\end{align*}
and then $\C_{D}$ is distance-optimal if
$r+\lfloor \frac{2}{q}\rfloor>1$.
This completes the proof.
\end{proof}

\begin{example}
Let $q=2$, $m=6$, $r=2$, $\theta_{1}=1$ and $\theta_{2}=\alpha$, where $\alpha$ is a primitive element of $\fqm$. Magma experiments show that $\C_{D}$ is a $[57,6,28]$ binary linear code with the weight enumerator $1+36z^{28}+24z^{30}+3z^{32}$, which is consistent with our result in Theorem \ref{optimalcode-con-4-1}. This code is distance-optimal due to \cite{GMB}.
\end{example}

\begin{example}
Let $q=3$, $m=4$, $r=2$, $\theta_{1}=1$ and $\theta_{2}=\alpha$, where $\alpha$ is a primitive element of $\fqm$. Magma experiments show that $\C_{D}$ is a $[64,4,42]$ linear code with the weight enumerator $1+64z^{42}+16z^{48}$, which is consistent with our result in Theorem \ref{optimalcode-con-4-1}. This code is distance-optimal due to \cite{GMB}.
\end{example}

The weight distribution of $\C_{D}$ in Theorem \ref{optimalcode-con-4} can be determined for $h=3$ as follow.

\begin{thm} \label{optimalcode-con-4-2}
Let $\C_{D}$ be defined by \eqref{CD} and \eqref{CD1-D4} with $h=3$ and $q^{m-r}>3$. Then $\C_{D}$ is a $[q^m-3q^r+2,m,(q-1)(q^{m-1}-3q^{r-1})]$ linear code with the following weight distribution
\begin{center}
\begin{tabular}{lll}
\hline weight $w$ & Multiplicity $A_{w}$\\ \hline
                                           $0$ & $1$ \\
$(q-1)q^{m-1}$               & $q^{m-2r}-1$ \\
$(q-1)(q^{m-1}-2q^{r-1})$ & $3q^{m-2r}(q^{r}-1)$ \\
$(q-1)(q^{m-1}-3q^{r-1})$ & $q^m-3q^{m-r}+2q^{m-2r}$ \\
\hline
\end{tabular}
\end{center}
if $\frac{\theta_{3}/\theta_{1}-(\theta_{3}/\theta_{1})^{q^r}}{\theta_{2}/\theta_{1}-(\theta_{2}/\theta_{1})^{q^r}} \in \fqr$, and for $\frac{\theta_{3}/\theta_{1}-(\theta_{3}/\theta_{1})^{q^r}}{\theta_{2}/\theta_{1}-(\theta_{2}/\theta_{1})^{q^r}} \not\in \fqr$, its weight distribution is given by
\begin{center}
\begin{tabular}{lll}
\hline weight $w$ & Multiplicity $A_{w}$\\ \hline
                                           $0$ & $1$ \\
$(q-1)q^{m-1}$               & $q^{m-3r}-1$ \\
$(q-1)(q^{m-1}-q^{r-1})$ & $3(q^r-1)q^{m-3r}$ \\
$(q-1)(q^{m-1}-2q^{r-1})$ & $3q^{m-r}-6q^{m-2r}+3q^{m-3r}$ \\
$(q-1)(q^{m-1}-3q^{r-1})$ & $q^m-3q^{m-r}+3q^{m-2r}-q^{m-3r}$ \\
\hline
\end{tabular}
\end{center}
Moreover, $\C_{D}$ is a near Griesmer code when $q=2$, and it is distance-optimal if $r>1$ (resp. $r>2$) when $q=2,\,3$ (resp. $q>3$).
\end{thm}

\begin{proof}
According to the proof of Theorem \ref{optimalcode-con-4}, one can obtain the parameters of $\C_{D}$ and conclude that $\C_{D}$ has the following four possible nonzero weights:
$w_{1}=(q-1)q^{m-1}$, $w_{2}=(q-1)(q^{m-1}-q^{r-1})$, $w_{3}=(q-1)(q^{m-1}-2q^{r-1})$ and $w_{4}=(q-1)(q^{m-1}-3q^{r-1})$.
Further, one can also have that $wt(c_{a})=w_{1}$ if and only if $\Tr_{q^{r}}^{q^{m}}(a\theta_{1})=\Tr_{q^{r}}^{q^{m}}(a\theta_{2})=\Tr_{q^{r}}^{q^{m}}(a\theta_{3})=0$ and $wt(c_{a})=w_{2}$ if and only if exactly one value in the set $$\{\Tr_{q^{r}}^{q^{m}}(a\theta_{1}),\Tr_{q^{r}}^{q^{m}}(a\theta_{2}),\Tr_{q^{r}}^{q^{m}}(a\theta_{3})\}$$ is not equal to $0$.

Similar to the calculation of $A_{w_{1}}$ in Theorem \ref{optimalcode-con-4-1}, for $h=3$, one can derive
\begin{align*}
A_{w_{1}}=&\left\{\begin{array}{ll}
q^{m-2r}-1, & \mbox{if $\frac{\theta_{3}/\theta_{1}-(\theta_{3}/\theta_{1})^{q^r}}{\theta_{2}/\theta_{1}-(\theta_{2}/\theta_{1})^{q^r}} \in \fqr$},  \\
q^{m-3r}-1, & \mbox{otherwise}  \\
\end{array}\right.
\end{align*}
by using the same manner.

Next, we calculate $A_{w_{2}}$ for $h=3$. Assume that $\{i,j,k\}=\{1,2,3\}$. Since $\theta_{k}/\theta_{i},\,\, \theta_{j}/\theta_{i} \notin \fqr$, then by employing the same technique, we have
\begin{align*}
&|\{x\in \fqm: \Tr_{q^{r}}^{q^{m}}(x\theta_{i})=0\,\,
{\rm and}\,\,\Tr_{q^{r}}^{q^{m}}(x\theta_{j})=0\,\,{\rm and}\,\,\Tr_{q^{r}}^{q^{m}}(x\theta_{k})\ne 0\}| \\
=&\frac{1}{q^{2r}}|\{x\in \fqm: \Tr_{q^{r}}^{q^{m}}(\frac{\theta_{k}/\theta_{i}-(\theta_{k}/\theta_{i})^{q^r}}{\theta_{j}/\theta_{i}-(\theta_{j}/\theta_{i})^{q^r}}
(x^{q^r}-x))\ne 0\}|\\
=&\left\{\begin{array}{ll}
0, & \mbox{if $\frac{\theta_{k}/\theta_{i}-(\theta_{k}/\theta_{i})^{q^r}}{\theta_{j}/\theta_{i}-(\theta_{j}/\theta_{i})^{q^r}} \in \fqr$},  \\
(q^r-1)q^{m-3r}, & \mbox{otherwise}
\end{array}\right.
\end{align*}
 which implies that
\begin{align*}
A_{w_{2}}=\left\{\begin{array}{ll}
0, & \mbox{if $\frac{\theta_{3}/\theta_{1}-(\theta_{3}/\theta_{1})^{q^r}}{\theta_{2}/\theta_{1}-(\theta_{2}/\theta_{1})^{q^r}} \in \fqr$},  \\
3(q^r-1)q^{m-3r}, & \mbox{otherwise}
\end{array}\right.
\end{align*}
since $\delta_{1}:=\frac{\theta_{3}/\theta_{1}-(\theta_{3}/\theta_{1})^{q^r}}{\theta_{2}/\theta_{1}-(\theta_{2}/\theta_{1})^{q^r}} \in \fqr$,  $\frac{\theta_{2}/\theta_{1}-(\theta_{2}/\theta_{1})^{q^r}}{\theta_{3}/\theta_{1}-(\theta_{3}/\theta_{1})^{q^r}} \in \fqr$ and  $\delta_{2}:=\frac{\theta_{1}/\theta_{2}-(\theta_{1}/\theta_{2})^{q^r}}{\theta_{3}/\theta_{2}-(\theta_{3}/\theta_{2})^{q^r}} \in \fqr$ hold simultaneously due to $\delta_{1}\in \fqr$ if and only if $\delta_{2}\in \fqr$. Note that $\delta_{1}\ne 0$ and $\delta_{2} \ne 0$ since $\theta_{i}/\theta_{j}\notin\fqr$ for any $1\leq i<j \leq 3$. It can be readily verified that $\delta_{2}(\theta_{3}/\theta_{1}-(\theta_{2}/\theta_{1})\delta_{1})=1$. If $\delta_{1}\in \fqr$, we have $(\frac{1}{\delta_{2}})^{q^{r}}-\frac{1}{\delta_{2}}=(\theta_{3}/\theta_{1})^{q^{r}}-\theta_{3}/\theta_{1}-
((\theta_{2}/\theta_{1})^{q^{r}}-\theta_{2}/\theta_{1})\delta_{1}=0$ which implies $\delta_{2}\in \fqr$. Similarly we have $\delta_{1}\in \fqr$ if $\delta_{2}\in \fqr$.
Then, the values of $A_{w_{3}}$ and $A_{w_{4}}$ follow from the first two Pless Power Moments.

By Theorem \ref{optimalcode-con-4}, $\C_{D}$ is a near Griesmer code when $q=2$. A straightforward calculation gives
\begin{align*}
g(m,d+1)=\sum_{i=0}^{m-1} \lceil \frac{(q-1)(q^{m-1}-3q^{r-1})+1}{q^i}\rceil
=\left\{\begin{array}{ll}
q^{m}-3q^{r}+r+1, & \mbox{if $q=2 \,\,{\rm or\,\,}3$},  \\
q^{m}-3q^{r}+r, & \mbox{if $q>3$}.  \\
\end{array}\right.
\end{align*}
Thus, when $q=2$ or $q=3$, $\C_{D}$ is distance-optimal if $r>1$ and when $q>3$ it is distance-optimal if $r>2$.
This completes the proof.
\end{proof}

\begin{example}
Let $q=2$, $m=6$, $r=2$, $\theta_{1}=1$, $\theta_{2}=\alpha$ and $\theta_{3}=1+\alpha$, where $\alpha$ is a primitive element of $\fqm$. Magma experiments show that $\C_{D}$ is a $[54,6,26]$ binary linear code with the weight enumerator $1+24z^{26}+36z^{28}+3z^{32}$, which is consistent with our result in Theorem \ref{optimalcode-con-4-2}. This code is a near Griesmer code and it is distance-optimal due to \cite{GMB}.
\end{example}

\begin{example}
Let $q=3$, $m=8$, $r=2$, $\theta_{1}=1$, $\theta_{2}=\alpha$ and $\theta_{3}=\alpha^2$, where $\alpha$ is a primitive element of $\fqm$. Magma experiments show that $\C_{D}$ is a $[6536,8,4356]$ linear code with the weight enumerator $1+4608z^{4356}+1728z^{4362}+216z^{4368}+8z^{4374}$, which is consistent with our result in Theorem \ref{optimalcode-con-4-2}. This code is distance-optimal due to the Griesmer bound.
\end{example}

\section{The fourth family of optimal linear codes} \label{section-6}

In this section, we study the linear codes $\C_{D}$ of the form \eqref{CDbi} with the defining set
\begin{eqnarray}\label{CD2-D2}
D=\{(x,y):x\in \fqm\backslash \fqr, y \in \fqk\backslash \fqs \},
\end{eqnarray}
where $m$, $k$, $r<m$, $s<k$ are positive integers satisfying $r|m$, $s|k$.

\begin{thm} \label{optimalcode-con-2}
Let $\C_{D}$ be defined by \eqref{CDbi} and \eqref{CD2-D2}. If $m+s\geq k+r$ and $q^{m-r}>q^{m-r+s-k}+1$, then \\
$1).$ $\C_{D}$ is a $[(q^{m}-q^{r})(q^{k}-q^{s}),m+k,(q-1)(q^{m+k-1}-q^{m+s-1}-q^{k+r-1})]$ linear code;\\
$2).$ $\C_{D}$ is $4$-weight with the following weight distribution:
\begin{center}
\begin{tabular}{llll}
\hline weight $w$ & Multiplicity $A_{w}$\\ \hline
                                           $0$ & $1$ \\
$(q-1)(q^{m+k-1}-q^{k+r-1})$               & $q^{k-s}-1$ \\
$(q-1)(q^{m+k-1}-q^{m+s-1})$ & $q^{m-r}-1$ \\
$(q-1)(q^{m+k-1}-q^{m+s-1}-q^{k+r-1})$ & $(q^{k-s}-1)(q^{m-r}-1)$ \\
$(q-1)(q^{m+k-1}-q^{m+s-1}-q^{k+r-1}+q^{r+s-1})$ & $q^{m+k}-q^{m+k-r-s}$ \\
\hline
\end{tabular}
\end{center}
$3).$ $\C_{D}$ is distance-optimal if $k+r>q^{r+s}$ (resp. $1+k+r>q^{r+s}$) when $m+s = k+r$ and $q\not=2$ (resp. $m+s \ne k+r$ or $q=2$); \\
$4).$ $\C_{D}$ is self-orthogonal if $(q,r+s)\notin \{(2,1),(2,2),(3,1)\}$;\\
$5).$ $\C_{D}$ is minimal if $q^{m+k}> q^{m+s+1}+q^{k+r}$.
\end{thm}

\begin{proof}
It is obvious that the length of $\C_{D}$ is $n=|D|=(q^{m}-q^{r})(q^{k}-q^{s})$.
For $(a,b)\not =(0,0)$, the Hamming weight $wt(c_{a,b})$ of the codeword $c_{a,b}$ in $\C_{D}$ is $n-N_{a,b}$, where $N_{a,b}=|\{(x,y) \in (\fqm\backslash \fqr) \times (\fqk\backslash \fqs): \Tr_{q}^{q^{m}}(ax)+\Tr_{q}^{q^{k}}(by) = 0\}|$. Using the orthogonal property of nontrivial additive characters gives
\begin{eqnarray*}
N_{a,b}&=&\frac{1}{q}\sum_{x\in \fqm\backslash \fqr}\sum_{y\in \fqk\backslash \fqs}
\sum_{u\in \fq} \chi(u(\Tr_{q}^{q^{m}}(ax)+\Tr_{q}^{q^{k}}(by)))\\
&=&\frac{1}{q}\sum_{u\in \fq} \sum_{x\in \fqm\backslash \fqr}\chi(u\Tr_{q}^{q^{m}}(ax))\sum_{y\in \fqk\backslash \fqs} \chi(u\Tr_{q}^{q^{k}}(by))
\end{eqnarray*}
which can be further expressed as
\begin{align*}
N_{a,b}=&\frac{1}{q}\sum_{u\in \fq}
(\sum_{x\in \fqm}\chi(u\Tr_{q}^{q^{m}}(ax))
\sum_{y\in \fqk}\chi(u\Tr_{q}^{q^{k}}(by))
-\sum_{x\in \fqm}\chi(u\Tr_{q}^{q^{m}}(ax))
\sum_{y\in \fqs}\chi(u\Tr_{q}^{q^{k}}(by))\\
&-\sum_{x\in \fqr}\chi(u\Tr_{q}^{q^{m}}(ax))
\sum_{y\in \fqk}\chi(u\Tr_{q}^{q^{k}}(by))
+\sum_{x\in \fqr}\chi(u\Tr_{q}^{q^{m}}(ax))
\sum_{y\in \fqs}\chi(u\Tr_{q}^{q^{k}}(by))).
\end{align*}
Note that
\[\frac{1}{q}\sum_{u\in \fq}\sum_{x\in \fqm}\chi(u\Tr_{q}^{q^{m}}(ax))\sum_{y\in \fqk}\chi(u\Tr_{q}^{q^{k}}(by))=q^{m+k-1}\]
holds for $(a,b)\not=(0,0)$. Then, it can be readily verified that
\begin{align*}
N_{a,b}=&\left\{\begin{array}{llll}
q^{m+k-1}-q^{m+s}-q^{k+r-1}+q^{r+s}, & \mbox{if}\,\,a=0,\,\,b \not=0,\,\, \Tr_{q^{s}}^{q^{k}}(b)=0, \\
q^{m+k-1}-q^{m+s-1}-q^{k+r}+q^{r+s}, & \mbox{if}\,\, a\not =0,\,\,b =0,\,\,\Tr_{q^{r}}^{q^{m}}(a)=0, \\
q^{m+k-1}-q^{m+s-1}-q^{k+r-1}+q^{r+s}, & \mbox{if}\,\, a\not =0,\,\,b \not=0,\,\,\Tr_{q^{r}}^{q^{m}}(a)=\Tr_{q^{s}}^{q^{k}}(b)=0, \\
q^{m+k-1}-q^{m+s-1}-q^{k+r-1}+q^{r+s-1},  & \mbox{otherwise}.
\end{array}\right.
\end{align*}
Consequently,  for $(a,b)\not =(0,0)$, $wt(c_{a,b})$ is equal to
\begin{align*} 
&\left\{\begin{array}{llll}
w_1:=(q-1)(q^{m+k-1}-q^{k+r-1}), & \mbox{if}\,\,a=0,\,\,b \not=0,\,\, \Tr_{q^{s}}^{q^{k}}(b)=0, \\
w_2:=(q-1)(q^{m+k-1}-q^{m+s-1}), & \mbox{if}\,\, a\not =0,\,\,b =0,\,\,\Tr_{q^{r}}^{q^{m}}(a)=0, \\
w_3:=(q-1)(q^{m+k-1}-q^{m+s-1}-q^{k+r-1}), & \mbox{if}\,\, a\not =0,\,\,b \not=0,\,\, \Tr_{q^{r}}^{q^{m}}(a)=\Tr_{q^{s}}^{q^{k}}(b)=0, \\
w_4:=(q-1)(q^{m+k-1}-q^{m+s-1}-q^{k+r-1}+q^{r+s-1}),  & \mbox{otherwise}.
\end{array}\right.
\end{align*}
Observe that $d=w_{3}=(q-1)q^{k+r-1}(q^{m-r}-q^{m-r+s-k}-1)>0$ since $q^{m-r}>q^{m-r+s-k}+1$.
This shows that the dimension of $\C_{D}$ is equal to $m+k$. According to the balanced property of trace functions, we have
$A_{w_{1}}=q^{k-s}-1$, $A_{w_{2}}=q^{m-r}-1$ and $A_{w_{3}}=(q^{k-s}-1)(q^{m-r}-1)$, which leads to $A_{w_{4}}=q^{m+k}-q^{m+k-r-s}$ due to $A_{w_{1}}+A_{w_{2}}+A_{w_{3}}+A_{w_{4}}=q^{m+k}-1$.

By using the Griesmer bound, one can obtain
\begin{align*} 
g(m+k,d)=
\left\{\begin{array}{ll}
q^{m+k}-q^{m+s}-q^{k+r}+1, & \mbox{if}\,\,m+s\not = k+r,  \\
q^{m+k}-q^{m+s}-q^{k+r},  & \mbox{if}\,\,m+s = k+r
\end{array}\right.
\end{align*}
and
\begin{align*}
g(m+k,d+1)=
\left\{\begin{array}{ll}
q^{m+k}-q^{m+s}-q^{k+r}+k+r,  & \mbox{if}\,\,m+s = k+r \,\,{\rm and}\,\, q\not=2,  \\
q^{m+k}-q^{m+s}-q^{k+r}+k+r+1,  & \mbox{otherwise}.
\end{array}\right.
\end{align*}
Thus, when $m+s=k+r$ and $q\not=2$, $\C_{D}$ is distance-optimal if $k+r>q^{r+s}$; and when $m+s \ne k+r$ or $q=2$, $\C_{D}$ is distance-optimal if
$1+k+r>q^{r+s}$.

To prove the self-orthogonality of $\C_D$, define $D_{1}=\fqm \times \fqk$, $D_{2}=\fqm \times \fqs$, $D_{3}=\fqr \times \fqk$ and $D_{4}=\fqr \times \fqs$. Let $\C_{D_{1}}$, $\C_{D_{2}}$, $\C_{D_{3}}$ and $\C_{D_{4}}$ be defined by \eqref{CDbi}, then similar to the proof of Lemma \ref{self-orthogonalF}, the linear codes $\C_{D_{1}}$, $\C_{D_{2}}$, $\C_{D_{3}}$ and $\C_{D_{4}}$ are self-orthogonal if $(q,r+s)\notin \{(2,1),(2,2),(3,1)\}$. Note that $D_{1}=D \cup (D_{2}\backslash D_{4})\cup (D_{3}\backslash D_{4})\cup D_{4}$. Thus, by Lemma \ref{self-orthogonalD-}, one can conclude that  $\C_{D}$ is self-orthogonal if $(q,r+s)\notin \{(2,1),(2,2),(3,1)\}$.

The minimality of $\C_{D}$ can be easily verified by Lemma \ref{minimal}. This completes the proof.
\end{proof}

\begin{remark}
Note that $\C_{D}$ in Theorem \ref{optimalcode-con-2} is reduced to a $3$-weight linear code if $m+s=k+r$.
\end{remark}


\begin{example}
Let $q=2$, $m=4$, $k=3$ and $r=s=1$. Magma experiments show that $\C_{D}$ is a $[84,7,40]$ binary linear code with the weight enumerator $1+21z^{40}+96z^{42}+7z^{48}+3z^{56}$, which is consistent with our result in Theorem \ref{optimalcode-con-2}. This code is distance-optimal due to \cite{GMB}.
\end{example}

\begin{example}
Let $q=2$, $m=4$, $k=4$ and $r=s=1$. Magma experiments show that $\C_{D}$ is a $[196,8,96]$ binary linear code with the weight enumerator $1+49z^{96}+192z^{98}+14z^{112}$, which is consistent with our result in Theorem \ref{optimalcode-con-2}. This code is distance-optimal due to \cite{GMB}.
\end{example}

Specially, if one takes $r=s=0$ and defines ${\mathbb F}_{q^{0}}=\{0\}$, then good codes can also be obtained as in Theorem \ref{optimalcode-con-2}. The proof is similar to that of Theorem \ref{optimalcode-con-2} and we omit it here.

\begin{thm} \label{optimalcode-3}
Let $m$, $k$ be positive integers with $k\leq m$ and $q^m>q^{m-k}+1$. Let $\C_{D}$ be defined by \eqref{CDbi} and $D= \{(x,y):x \in \fqm^*, y \in \fqk^* \}$. Then \\
$1).$ $\C_{D}$ is a $[q^{m+k}-q^{m}-q^{k}+1,m+k,(q-1)(q^{m+k-1}-q^{m-1}-q^{k-1})]$ linear code;\\
$2).$ $\C_{D}$ is $3$-weight with the following weight distribution:
\begin{center}
\begin{tabular}{llll}
\hline weight $w$ & Multiplicity $A_{w}$\\ \hline
                                           $0$ & $1$ \\
$(q-1)(q^{m+k-1}-q^{k-1})$               & $q^{k}-1$ \\
$(q-1)(q^{m+k-1}-q^{m-1})$ & $q^{m}-1$ \\
$(q-1)(q^{m+k-1}-q^{m-1}-q^{k-1})$ & $q^{m+k}-q^{m}-q^{k}+1$ \\
\hline
\end{tabular}
\end{center}
$3).$ $\C_{D}$ is a Griesmer code when $m\not=k$ and it is a near Griesmer code ${\rm(}$also distance-optimal if $m+\lfloor \frac{2}{q}\rfloor>1$${\rm)}$ when $m=k$;\\
$4).$ $\C_{D}$ is self-orthogonal if $(q,k)\notin \{(2,1),(2,2),(3,1)\}$;\\
$5).$ $\C_{D}$ is minimal if $q^{m+k}> q^{m+1}+q^{k}$.
\end{thm}

\begin{remark}
The linear code $\C_{D}$ in Theorem \ref{optimalcode-3} is reduced to a 2-weight linear code when $m=k$ and the Griesmer code in Theorem \ref{optimalcode-3} has the same parameters with the Solomon and Stiffler code in the nonprojective case.
\end{remark}

\begin{example}
Let $q=2$, $m=5$, $k=4$. Magma experiments show that $\C_{D}$ is a $[465,9,232]$ binary linear code with the weight enumerator $1+465z^{232}+31z^{240}+15z^{248}$, which is consistent with our result in Theorem \ref{optimalcode-3}. This code is a Griesmer code.
\end{example}

\begin{example}
Let $q=2$, $m=4$, $k=4$. Magma experiments show that $\C_{D}$ is a $[225,8,112]$ binary linear code with the weight enumerator $1+225z^{112}+30z^{120}$, which is consistent with our result in Theorem \ref{optimalcode-3}. This code is a near Griesmer code and it is distance-optimal due to \cite{GMB}.
\end{example}

\begin{remark}
It is known that equivalent codes have the same parameters and weight distribution, but the converse is not necessarily true, so it is normally difficult to discuss the equivalence of codes.  In 2020, several infinite families of optimal binary linear codes of the form
$\C_{P}=\{c_{a}=(a\cdot x)_{x\in P}: a\in \ftwo^{m}\}$
were presented in \cite{JYHLL}, where $P=\ftwo^{m}\setminus \Delta$ and $\Delta$ is a simplicial complex in $\ftwo^{m}$.
Firstly, our results in this paper holds for a prime power $q$, thus our codes are new when $q>2$. Secondly, for $q=2$ and $m=4,5$, when $\Delta$ runs through all simplicial complexes in $\ftwo^{4}$ and $\ftwo^{5}$, computer experiments show that linear codes with new parameters can be produced in Sections \ref{section-4} and \ref{section-5} by comparing with all the linear codes $\C_{P}$ obtained in \cite{JYHLL}. Our codes in Section \ref{section-3} have the same parameters as those in \cite{JYHLL} and it should be noted that whether the set $\Omega_1=\cup_{i=1}^{h} {\mathbb F}_{2^{r_{i}}}$ in Section \ref{section-3} is a simplicial complex in $\ftwo^{m}$ depends on the selected basis of $\ftwom$ over $\ftwo$. For $q=2$ and the linear codes in Section \ref{section-6}, the parameters of our codes in Theorem \ref{optimalcode-con-2} are the same with the codes in
\cite[Theorem V.2(iii)]{JYHLL} and the parameters of our codes in Theorem \ref{optimalcode-3} are different from those of \cite[Theorem V.2]{JYHLL}.
 \end{remark}

\section{Conclusions}

The construction of optimal linear codes is a hard problem.
In this paper, we constructed four families of linear codes over finite fields via the defining set approach, which can produce infinite families of optimal linear codes, including infinite families of (near) Griesmer codes. Using the Griesmer bound, we characterized the optimality of
these four families of linear codes with an explicit computable criterion and obtained infinite families of distance-optimal linear codes. Moreover, we obtained several classes of distance-optimal linear codes with few weights and completely determined their weight distributions. In addition, we investigated the self-orthogonality and minimality of these linear codes and it is shown that most of them are either self-orthogonal or minimal.

\section*{Acknowledgements}

This work was supported in part by the National Natural Science Foundation of China (Nos. 62072162, 61761166010, 12001176, 61702166), the Application Foundation Frontier Project of Wuhan Science and Technology Bureau (No. 2020010601012189) and the National Key Research and Development Project (No. 2018YFA0704702).


\begin{thebibliography}{99}


\bibitem{ARJD} R. J. Anderson, C. Ding, T. Hellsesth, T. Klove, How to build robust shared control systems, Des. Codes Cryptography 15(2) (1998), pp.  111-123.

\bibitem{AAAB} A. Ashikhmin, A. Barg, Minimal vectors in linear codes, IEEE Trans. Inf. Theory 44(5) (1998), pp. 2010-2017.

\bibitem{CAGJ} A. R. Calderbank, J. Goethals, Three-weight codes and association schemes, Philips J. Res. 39(4-5) (1984), pp. 143-152.

\bibitem{CRSS} A. R. Calderbank, E. M. Rains, P. W. Shor, N. J. A. Sloane,
Quantum error correction and orthogonal geometry, Phys. Rev. Lett 78 (1997), pp. 405-408.

\bibitem{CCDY} C. Carlet, C. Ding, J. Yuan,  Linear codes from perfect nonlinear mappings and their secret sharing schemes,  IEEE Trans. Inf. Theory 51(6) (2005), pp. 2089-2102.

\bibitem{DING1} C. Ding, Codes from Difference Sets, World Scientific, Singapore (2015).

\bibitem{DING2} C. Ding, Designs from Linear Codes, World Scientific, Singapore (2018).

\bibitem{DCHT} C. Ding, T. Helleseth, T. Kl{\o}ve, X. Wang, A generic construction of cartesian authentication codes, IEEE Trans. Inf. Theory 53(6) (2007), pp. 2229-2235.

\bibitem{DN} C. Ding, H. Niederreiter, Cyclotomic linear codes of order $3$, IEEE Trans. Inf. Theory 53(6) (2007), pp. 2274-2277.

\bibitem{DW} C. Ding, X. Wang, A coding theory construction of new systematic authentication codes,  Theor. Comput. Sci. 330(1) (2005), pp. 81-99.

\bibitem{DY} C. Ding, J. Yuan, Covering and secret sharing with linear codes,
in: Discrete Mathematics and Theoretical Computer Science, LNCS 2731 (2003), Springer Verlag, pp. 11-25.

\bibitem{DGOT} D. Gottesman, A class of quantum error-correcting codes saturating the quantum Hamming bounds, Phys. Rev. A 54 (1996), pp. 1862-1868.

\bibitem{GMB}M. Grassl, Bounds on the minimum distance of linear codes, Online available at http://www.codetables.de, Accessed on 2021-08-11.

\bibitem{JHG} J.H. Griesmer, A bound for error correcting codes, IBM J. Res. Dev. 4 (1960), pp. 532-542.

\bibitem{HELL} T. Helleseth, Projective codes meeting the Griesmer bound, Discrete Math. 106/107 (1992), pp. 265-271.

\bibitem{HDWZ} Z. Heng, C. Ding, W. Wang, Optimal binary linear codes from maximal arcs, IEEE Trans. Inf. Theory 66(9) (2020), pp. 5387-5394.

\bibitem{HWPV} W. Huffman, V. Pless, Fundamentals of error-correcting codes, Cambridge University Press (1997).

\bibitem{JYHLL} J. Y. Hyun, J. Lee, Y. Lee, Infinite families of optimal linear codes constructed from simplicial complexes, IEEE Trans. Inf. Theory 66(11) (2020), pp. 6762-6773.

\bibitem{KKKS} A. Ketkar, A. Klappenecker, S. Kumar, P. K. Sarvepalli, Nonbinary stabilizer codes over finite fields, IEEE Trans. Inf. Theory 52(11) (2006), pp. 4892-4914.

\bibitem{Lidl} R. Lidl, H. Niederreiter,  Finite Fields, Encyclopedia of Mathematics, vol. 20, Cambridge University Press, Cambridge (1983).

\bibitem{JLM}  J. L. Massey, Minimal codewords and secret sharing, in: Proc. 6th Joint
Swedish-Russian Workshop on Information Theory (Molle, Sweden, 1993), pp. 246-249.

\bibitem{GSJS} G. Solomon, J.J. Stiffer, Algebraically punctured cyclic codes, Inform. and Control 8 (1965), pp. 170-179.

\bibitem{WZQY} Y. Wu, X. Zhu, Q. Yue, Optimal few-weight codes from simplicial complexes, IEEE Trans. Inf. Theory 66(6) (2020), pp. 3657-3663.

\bibitem{JYCD}  J. Yuan, C. Ding, Secret sharing schemes from three classes of
linear codes, IEEE Trans. Inf. Theory 52(1) (2006), pp. 206-212.
\end{thebibliography}
\end{document}